\documentclass[10pt]{article}
\usepackage{graphicx,multicol}

\usepackage[a4paper, left=1.0in, right=1.0in]{geometry} 

\usepackage[latin1]{inputenc}

\usepackage{color}
\usepackage{amsthm,amsmath}
\usepackage{amssymb}
\usepackage{amsfonts}
\usepackage{hyperref}
\hypersetup{colorlinks,allcolors=black}
\usepackage{graphicx}
\usepackage{tabularx}
\usepackage{orcidlink}
\usepackage{lscape}
\usepackage{pdflscape}
\usepackage{adjustbox} 
\usepackage{array} 

\newtheorem{theorem}{Theorem}
\newtheorem{definition}{Definition}
\newtheorem{corollary}{Corollary}
\newtheorem{proposition}{Proposition}
\newtheorem{lemma}{Lemma}
\newtheorem{example}{Example}
\newtheorem{remark}{Remark}

\newtheorem{Fact}{Fact}
\newtheorem{prob}{Open Problem}

\date{ }

\author{Atif Ahmad Khan$^{1,}\footnote{Corresponding author}$ \orcidlink{0009-0009-2371-5707}~, Shakir Ali$^{1}$ \orcidlink{0000-0001-5162-7522} , Elif Segah Oztas$^2$ \orcidlink{0000-0003-0032-3400}, Abhishek Kesarwani$^3$ \orcidlink{
		0000-0001-8488-6811}\\
	\small{$^1$Department of Mathematics, Faculty of Science, }\\
	\small{Aligarh Muslim University, Aligarh 202002, India}\\
	\small{atifkhanalig1997@gmail.com, shakir.ali.mm@amu.ac.in}\\
	\small{$^2$Department of Mathematics, Karamanoglu Mehmetbey  University, Turkiye}\\
	\small{elifsegahoztas@gmail.com}\\ 
		\small{$^3$ Department of Mathematics \& Computing, Indian Institute of Information Technology Vadodara}, \\\small{International Campus Diu, Daman \& Diu 362520, India }\\
	\small{abhishek$\_$kesarwani@iiitvadodara.ac.in}\\ 
	\small{ }}

\begin{document}
\title{\textbf{MDS matrices from skew polynomials with automorphisms and derivations}}
\maketitle
\begin{abstract}
Maximum Distance Separable (MDS) matrices play a central role in coding theory and symmetric-key cryptography due to their optimal diffusion properties. In this paper, we present a construction of MDS matrices using skew polynomial rings \( \mathbb{F}_q[X;\theta,\delta] \), where \( \theta \) is an automorphism and \( \delta \) is a \( \theta\)-derivation on \( \mathbb{F}_q \). We introduce the notion of \( \delta_{\theta} \)-circulant matrices and study their structural properties. Necessary and sufficient conditions are derived under which these matrices are involutory and satisfy the MDS property. The resulting $\delta_\theta$-circulant matrix can be viewed as a  generalization of classical constructions obtained in the absence of $\theta$-derivations.  One of the main contribution of this work is the construction of quasi recursive MDS matrices. In the setting of the skew polynomial ring $\mathbb{F}_q[X;\theta]$, we construct quasi recursive MDS matrices associated with companion matrices.
 These matrices are shown to be involutory, yielding a strict improvement over the quasi-involutory constructions previously reported in the literature. Several illustrative results and examples are also provided.

\end{abstract}


\section{Introduction}
The notions of confusion and diffusion in the design of encryption systems were first introduced by Shannon in his seminal paper ``Communication Theory of Secrecy Systems''~\cite{shannon1949communication}. The purpose of the confusion layer is to obscure the relationship between the secret key and the ciphertext, whereas the diffusion layer aims to conceal the relationship between the plaintext and the ciphertext. When employed in an iterated block cipher, alternating layers of confusion and diffusion ensure that every bit of the plaintext and the secret key has a complex and non-linear influence on every bit of the ciphertext. Consequently, the design of the diffusion layer has a direct impact on both the security and the efficiency of cryptographic primitives. In particular, diffusion layers play a crucial role in providing resistance against differential and linear cryptanalysis, as emphasized in~\cite{joan2002design}. Among the quantitative measures of diffusion, the branch number~\cite{joan2002design} is a fundamental parameter for evaluating diffusion strength. Optimal diffusion can be achieved using multipermutations~\cite{schnorr1994black,vaudenay1994need} or through the use of Maximum Distance Separable (MDS) matrices~\cite{joan2002design}.

Due to their strong algebraic structure and favorable implementation properties, circulant and recursive MDS matrices, together with their various generalizations, form some of the most extensively studied families of MDS matrices in cryptography. Such matrices have been adopted in prominent cryptographic designs, including block ciphers such as AES~\cite{joan2002design} and SHARK~\cite{rijmen1996cipher}, as well as lightweight hash functions like PHOTON~\cite{guo2011photon}. A key advantage of these constructions is that they admit elegant algebraic descriptions in terms of the polynomial ring $\mathbb{F}_q[X]$ and its quotient structures. For further background and developments, we refer the reader to~\cite{gupta2017towards,cauchois2019circulant,cauchois2016direct} and the references therein. Several structured families of MDS matrices have been proposed in the literature with the objective of achieving strong diffusion while maintaining efficient implementations. Among these, circulant and recursive MDS matrices have received particular attention.

Circulant MDS matrices have been widely investigated due to their simple structure and efficient realizations in diffusion layers; see, for example,~\cite{adhiguna2022orthogonal,ali2024circulant,ali2024xor,cauchois2019circulant,gupta2015cryptographically}. In this direction, Gupta et al.~\cite{gupta2015cryptographically} established several non-existence results for circulant matrices of even order that are simultaneously MDS and involutory. Subsequently, Cauchois et al.~\cite{cauchois2019circulant} introduced an algebraic framework for circulant MDS matrices based on polynomial representations arising from skew polynomial rings. By generalizing classical circulant matrices through an automorphism $\theta$, they defined $\theta$-circulant  matrices and showed that this generalized structure admits  involutory MDS matrices of even order.

Another important line of research concerns recursive MDS matrices, which are particularly attractive for resource-constrained implementations. Such matrices have been employed in cryptographic primitives including the PHOTON hash family~\cite{guo2011photon} and the LED block cipher~\cite{guo2011led}. A recursive MDS matrix is one that can be expressed as a power of a companion matrix associated with a polynomial $g(X)\in\mathbb{F}_q[X]$, enabling highly efficient realizations using linear feedback shift registers and minimal memory requirements in software~\cite{guo2011photon,guo2011led,gupta2013mds}. Following the work of Guo et al., numerous constructions have been proposed using both ad-hoc search techniques and coding-theoretic approaches~\cite{augot2013exhaustive,augot2014direct,berger2013construction,gupta2013mds,sajadieh2012recursive,wu2012recursive,kesarwani2021recursive}. Notably, Berger~\cite{berger2004gabidulin} derived infinite families of recursive MDS matrices from Gabidulin codes, which also satisfy the MRD property, while Augot and Finiasz~\cite{augot2014direct} later proposed a direct construction based on shortened BCH codes. In 2016, Cauchois et al.~\cite{cauchois2016direct} investigated recursive and quasi recursive MDS matrices within the skew polynomial framework, and obtaining quasi-involutory MDS matrices over a restricted class of finite fields.


Our constructions are carried out in the setting of skew polynomial rings, which naturally arise from the presence of 
an automorphism and a $\theta$-derivation on the underlying coefficient ring. 
This class of noncommutative polynomial rings was first introduced by Ore \cite{ore1933theory}. 
Therefore, before proceeding further, we briefly review the notion of derivations. Jacobson, in his classical book ``Structure of Rings" \cite{Jacobson}, introduced the concept of $(\theta, \gamma)$-derivation. Specifically, a map $\delta: \mathfrak{R} \to \mathfrak{R}$ is called a $(\theta, \gamma)$-derivation on $\mathfrak{R}$ if it satisfies  $
\delta(ab) = \theta(a)\delta(b) + \delta(a)\gamma(b),$ 
for all $a, b \in \mathfrak{R}$, where $\theta$ and $\gamma$ are endomorphisms of $\mathfrak{R}$. If we take $\gamma$ as an identity map on $\mathfrak{R}$, then $\delta$ is called $\theta$-derivation.  
The theory of automorphisms and derivations plays a crucial role in both pure and applied mathematics (cf.; \cite{nurcan1987sigma}, \cite{ashraf2006derivations}, and  \cite{ Nowicki} ). The second author, in collaboration with  others, has studied the structure of rings and algebras through derivations (see  \cite{ali2011generalized},  \cite{ali2024certain} for more details). Recently, several mathematicians have applied the theory of derivations in coding theory and cryptography (viz.; \cite{boucher2014linear},  \cite{boulagouaz2013sigma},  and \cite{Ou-azzou2024Onthe}). Precisely,  Boucher and Ulmer \cite{boucher2014linear} initiate the use of theory of derivations. By proposing a new evaluation of skew polynomials, they developed new classes of evaluation codes. Additionally, they generalized Gabidulin's decoding algorithm for the rank metric and constructed families of codes that satisfy the MDS and Maximum Rank Distance (MRD) properties.


Inspired by the above-mentioned studies on circulant and recursive MDS matrices, we develop our constructions using skew polynomial rings. In this way, skew polynomial rings provide a single algebraic setting in which both types of matrices can be studied. Specifically, we work in the non-commutative ring $\mathbb{F}_q[X;\theta,\delta]$, where $\theta$ is an automorphism of the finite field $\mathbb{F}_q$ and $\delta$ is a $\theta$-derivation. 
The multiplication is determined by the rule $Xa=\theta(a)X+\delta(a)$. When restricted to finite fields, every $\theta$-derivation admits a simple characterization of the form $\delta(a)=\beta(\theta(a)-a)$ for some $\beta\in\mathbb{F}_q$~\cite[Chapter~8, Theorem~3.1]{Cohn1985free}, which allows explicit control over the resulting constructions. Motivated by earlier studies of \( (\theta,\delta) \)-codes and \( (\theta,\delta) \)-cyclic matrices~\cite{boulagouaz2013sigma,Ou-azzou2024Onthe}. We generalize existing \( \theta \)-circulant matrix constructions by incorporating a \( \theta \)-derivation. This extension gives rise to a new class of structured matrices, referred to as \( \delta_{\theta} \)-circulant matrices. These matrices are analyzed via their associated skew polynomials of the form
\[
h\langle X\rangle = X^{m}-1+\sum_{i=0}^{m-1} h_i X^i \in \mathbb{F}_q[X;\theta,\delta],
\]
under the condition \( m = 2^{t} \) for some \( t \in \mathbb{N} \). Explicit criteria on the generating polynomial are derived to ensure that the resulting matrices are involutory. Although involutory MDS matrices have already been observed within the skew polynomial framework. In particular, Cauchois et al.\ \cite{cauchois2019circulant} investigated $\theta$-circulant matrices and exhibited an example of order $4$ that is simultaneously involutory and MDS. No such example exists in the classical circulant setting. Motivated by this, we use \( \theta \)-derivation into the construction in order to enlarge the class of $\theta$-circulant matrices.

A central result of this paper is the analysis of quasi recursive MDS matrices derived from the skew polynomial framework. In contrast to previous constructions \cite[Theorem 2]{cauchois2016direct}, which were confined to specific classes of finite fields and produced quasi-involutory MDS matrices of even order, our approach applies to a broader field setting. We prove that, for an appropriate choice of the automorphism \( \theta \), the resulting quasi recursive construction yields  involutory MDS matrices of even order. This significantly strengthens existing results, as involutory MDS matrices are highly desirable in cryptographic applications due to their reduced implementation cost.

Finally, beyond these primary constructions, we derive additional MDS matrices from existing ones via suitable algebraic transformations, further illustrating the flexibility of the proposed framework. In particular, we use Hadamard product as an effective tool for generating new MDS matrices.

The paper is structured as follows: Section 2 provides an overview of $\theta$-derivations, skew polynomial rings, and MDS matrices. Section~3 focuses on the construction and analysis of $\delta_{\theta}$-circulant MDS matrices, including their involutory property. We conclude this section by formulating an open problem related to the complete characterization of such  $\delta_{\theta}$-circulant MDS matrices.
 In Section~4, we derive results showing that every quasi recursive MDS matrix yields further quasi recursive MDS matrices. Section 5 explores the construction of quasi recursive MDS matrices using companion matrices. Additionally, we provide some examples of MDS matrices to illustrate their properties and construction methods.

\section{\textbf{Notations and preliminaries}}

In this section, we provide some basic definitions and results that will be used to deduce further results. Now, we start with the definition of $\theta$-derivation over finite field. Let $\mathbb{F}_q$ be a finite field and $\theta$ be an automorphism of $\mathbb{F}_q$. A $\theta$-derivation is a map $\delta:\mathbb{F}_q\longrightarrow \mathbb{F}_q$ such that for all $a$, $b \in \mathbb{F}_q$,

$$\delta(a+b)=\delta(a)+\delta(b)$$

and
\begin{eqnarray*}
	\delta(ab)&=&\delta(a)b+\theta(a)\delta(b).
\end{eqnarray*}
\newpage
We use the following notations throughout the paper:
\begin{table}[h]
	\centering
	\begin{tabular}{c p{10cm}}
		\hline
		$n,~m,~l,~k,~d,~t$& positive integers.\\
		\hline
		$\mathbb{F}_{q}$ &  the finite field of characteristic 2 with $q$ elements, where $q$ is a  power of 2. \\
		\hline
		$\mathbb{F}^n_{q}$&  the linear space over $\mathbb{F}_{q}$ of dimension $n$.\\
		\hline
		$M(m,\mathbb{F}_{q})$ & the set of all square  matrices of order $m$  over   $\mathbb{F}_{q}$. \\
		\hline
		$GL(m,\mathbb{F}_q)$&  the general linear group consisting of all $m \times m$ invertible matrices over $\mathbb{F}_q$. \\
		\hline
		$P(m,~\mathbb{F}_{2})$ & the set of all permutation  matrices of order $m$  over   $\mathbb{F}_{2}$. \\ 
		\hline
		$\mathbb{F}_q[X;\theta, \delta]$& the skew polynomial ring, where $\theta:\mathbb{F}_q\longrightarrow \mathbb{F}_q$ is an automorphism and $\delta:\mathbb{F}_q\longrightarrow \mathbb{F}_q$ is a $\theta$-derivation. \\
		\hline
		$Q\langle X \rangle$ & a skew polynomial over $\mathbb{F}_{q}$. \\
		\hline
		\hline
	\end{tabular}
\end{table}\\
\begin{example}
	Let $\alpha$ be a root of the polynomial $1+X+X^4$ over $\mathbb{F}_2$. Then
	$\mathbb{F}_{2^4}=\mathbb{F}_2(\alpha)$ is a finite field of order $16$. Consider the maps
	$\theta,\delta:\mathbb{F}_{2^4}\to\mathbb{F}_{2^4}$ given by
	\[
	\theta(a)=a^2
	\quad\text{and}\quad
	\delta(a)=a^2-a.
	\]
	Then, $\delta$ is a $\theta$-derivation on $\mathbb{F}_{2^4}$.
\end{example}

The following fact illustrates that the automorphism \( \theta \) and the associated \( \theta \)-derivation \( \delta \) do not commute in general.

\begin{Fact}\label{5p1}  
	Let \(\delta: \mathbb{F}_{q} \longrightarrow \mathbb{F}_{q}\) be a \(\theta\)-derivation. Next, let \(\theta: \mathbb{F}_{q} \longrightarrow \mathbb{F}_{q}\) be a non-identity automorphism such that \(\delta(a) = \beta(\theta(a) - a)\) for all \(a \in \mathbb{F}_{q}\) with \(\theta(\beta) \neq \beta\). Then, there exists an element \(r \in \mathbb{F}_{q}\) such that  
	\[
	(\delta \circ\theta)(r) \neq (\theta \circ  \delta)(r).
	\]
\end{Fact}
\begin{example}
	Let $\alpha$ be a root of the polynomial $1+X+X^4$ over $\mathbb{F}_2$. Then
	$\mathbb{F}_{2^4}=\mathbb{F}_2(\alpha)$ is a finite field of order $16$. Consider the maps
	$\theta,\delta:\mathbb{F}_{2^4}\to\mathbb{F}_{2^4}$ given by
	\[
	\theta(a)=a^2
	\quad\text{and}\quad
	\delta(a)=\alpha(a^2-a).
	\]
	Then, $\delta$ is a $\theta$-derivation on $\mathbb{F}_{2^4}$.
\end{example}
\begin{definition}
	Let $\theta:\mathbb{F}_q\to \mathbb{F}_q$ be a field automorphism. The \emph{fixed field} of $\mathbb{F}_q$ under $\theta$ is defined by
	\[
	(\mathbb{F}_q)^{\theta}=\{\,a\in \mathbb{F}_q:\theta(a)=a\,\}.
	\]
\end{definition}

\begin{definition}\label{5d1}  
	Let \(\theta:\mathbb{F}_{q} \longrightarrow \mathbb{F}_{q}\) be an automorphism and \(\delta: \mathbb{F}_{q} \longrightarrow \mathbb{F}_{q}\) be a derivation. A skew polynomial ring over a field \(\mathbb{F}_{q}\) with automorphism \(\theta\) and $\theta$-derivation \(\delta\), denoted \(\mathbb{F}_{q}[X; \theta, \delta]\), is defined as the set  
	\[
	\mathbb{F}_{q}[X; \theta, \delta] := \left\{ a_0 + a_1 X + \dots + a_{n-1} X^{n-1} : a_i \in \mathbb{F}_{q}, \, n \in \mathbb{N} \right\},
	\]  
	equipped with the usual addition of polynomials and a multiplication followed by the rule  
	\begin{eqnarray*}
		X * a &=& \theta(a) X + \delta(a) \quad \text{for all } a \in \mathbb{F}_{q},\\
		a*X&=& aX \quad \text{for all } a \in \mathbb{F}_{q}.
	\end{eqnarray*}
	
\end{definition}

In relation to some well known  polynomial rings from the literature, we make the following remarks.
\begin{remark}\label{5Rem1}
	Note that the structure of the skew polynomial ring \(\mathbb{F}_{q}[X; \theta, \delta]\) (cf.~\cite{liu2014kotter}) varies depending on the properties of the automorphism \(\theta\) and the $\theta$-derivation \(\delta\). 
	\begin{enumerate}  
		\item  If \(\delta\) is the zero derivation, that is, \(\delta(a) = 0\) for all \(a \in \mathbb{F}_{q}\), the ring simplifies to \(\mathbb{F}_{q}[X; \theta]\), known as a \textit{twisted polynomial ring}.  
		
		\item  If \(\theta\) is the identity automorphism, that is, \(\theta(a) = a\) for all \(a \in \mathbb{F}_{q}\), the ring \(\mathbb{F}_{q}[X; \delta]\) is referred to as a \textit{differential polynomial ring}.  
		
		\item  If  \(\theta\) is the identity automorphism and \(\delta\) is the zero derivation, the structure reduces to the \textit{ordinary polynomial ring} \(\mathbb{F}_{q}[X]\).  
	\end{enumerate}
\end{remark}

Skew polynomial rings behave quite differently from ordinary polynomial rings, yet they retain the important property of having a division algorithm. In the following proposition, we collect the properties of division algorithm for skew polynomial ring.
\begin{proposition}\label{5pp1}\cite[Proposition 2]{liu2014kotter}
	\( \mathbb{F}_{q}[X; \theta, \delta] \) is a right Euclidean domain.
	
\end{proposition}

\noindent
As a consequence of the  Proposition \ref{5pp1}, for any polynomials \( f\langle X\rangle, g\langle X\rangle \in \mathbb{F}_{q}[X; \theta, \delta] \) with \( g\langle X\rangle \neq 0 \), there exists unique polynomials \( q\langle X\rangle \) (the quotient) and \( r\langle X\rangle \) (the remainder) in \( \mathbb{F}_{q}[X; \theta, \delta] \) such that
\begin{eqnarray}\label{5eqn111}
	f\langle X\rangle &=& q\langle X\rangle * g\langle X\rangle + r\langle X\rangle,
\end{eqnarray}

\noindent with either \( r\langle X\rangle = 0 \) or \(\deg(r\langle X\rangle) < \deg(g\langle X\rangle)\). If $r\langle X\rangle=0,$ then $f\langle X \rangle=q\langle X\rangle*g\langle X\rangle=0~\ \mod _*g.$\\

Linear codes over finite fields are a class of error-correcting codes, where each codeword is a vector in a vector space over a finite field \(\mathbb{F}_q\). A linear code $C$ is defined as a subspace of \(\mathbb{F}_q^n\). The code is typically denoted as \([n, m, d]_q\), where \(n\) is the length of  $C$, \(m\) is the dimension of  $C$, and \(d\) is the minimum Hamming distance between distinct codewords in $C$. When \(d = n - m + 1\), the code is called  MDS, and the redundant part of the generator matrix forms an MDS matrix. To move forward, we first present the following definitions and lemma, which play a key role in the subsequent discussion.
\begin{definition}\label{5def001}\cite[Definition 4]{gupta2017towards}
	An \([n, m, d]_q\) code $C$ with generator matrix \( G = [M|I] \), where \( M \) is a \( m \times (n - m) \) matrix, is MDS if and only if every square submatrix of \( M \) is nonsingular. We say that \( M \) is an MDS matrix if the corresponding code $C$ is MDS.
\end{definition}

\begin{lemma}\label{5l9}
	\cite[Corollary 3, p. 319]{MacWilliams1977the} Let \( M \in M(m,\mathbb{F}_q) \). Then, the matrix \( M \) is MDS if and only if any \( m \) columns of \( G = [M|I] \) are linearly independent over \( \mathbb{F}_q \). Furthermore, \( M \) is MDS if and only if any \( m \) rows of \( \bar{G} = \begin{bmatrix} I \\ M \end{bmatrix} \) are linearly independent over \( \mathbb{F}_q \).
\end{lemma}
\begin{lemma}
	Let \( C \) be an \([2m,m,d]\) MDS code over a finite field \( \mathbb{F}_q \). Then, the minimum Hamming weight of any nonzero codeword in \( C \) satisfies:
	\[
	\text{wt}(C) \geq m + 1.
	\]
\end{lemma}

From a cryptographic point of view, we require the inverse of a matrix; therefore, we restrict our attention to square matrices. In particular, we focus on constructing \([n=2m,\, m,\, d=m+1]_q\) codes for some positive integer \(m\). From these codes, we can derive MDS matrices of size \(m \times m\) over \(\mathbb{F}_q\).

\begin{definition}
	A square matrix \(A\) of order $m$ is said to be MDS if every square submatrix of \(A\) is nonsingular.
\end{definition}

\begin{definition}
	Let \(r\) be a positive integer. A matrix \(M\) is called recursive MDS or \(r\)-MDS if the matrix \(M^r\) is MDS. If \(M\) is \(r\)-MDS, we say that \(M\) yields an MDS matrix.
\end{definition}

	\section{\textbf{ $\delta_\theta$-Circulant MDS matrices}}

In this section, following the algebraic framework developed by Cauchois et al.~\cite{cauchois2019circulant} for circulant and \(\theta\)-circulant matrices. We extend their work to the more general setting of the skew polynomial ring \(\mathbb{F}_q[X;\theta,\delta]\). We begin by introducing the definition of a \(\delta_{\theta}\)-circulant matrix.

\begin{definition}\label{5d1}
	Let \(\theta : \mathbb{F}_q \longrightarrow \mathbb{F}_q\) be an automorphism and \(\delta : \mathbb{F}_q \longrightarrow \mathbb{F}_q\) be a \(\theta\)-derivation. A matrix 
	\[
	A = \begin{bmatrix}
		a_{0,0} & a_{0,1} & \cdots & a_{0,m-1} \\
		a_{1,0} & a_{1,1} & \cdots & a_{1,m-1} \\
		\vdots & \vdots & \ddots & \vdots \\
		a_{j,0} & a_{j,1} & \cdots & a_{j,m-1} \\
		a_{j+1,0} & a_{j+1,1} & \cdots & a_{j+1,m-1} \\
		\vdots & \vdots & \ddots & \vdots \\
		a_{m-1,0} & a_{m-1,1} & \cdots & a_{m-1,m-1}
	\end{bmatrix}_{m\times m},
	\]
	is called a \(\delta_\theta\)-circulant matrix if the \((j+1)\)$^{th}$ row of \(A\) can be written as 
	\[
	(\delta(a_{j,0}) + \theta(a_{j,m-1}), \delta(a_{j,1}) + \theta(a_{j,0}), \dots, \delta(a_{j,m-1}) + \theta(a_{j,m-2})),
	\]
	for \(1 \leq j \leq m-1\). We denote this \(\delta_\theta\)-circulant matrix as \(\delta_\theta\)-circ\((a_{0,0}, a_{0,1}, \dots, a_{0,m-1})\).
\end{definition}

\noindent The $\delta_\theta$-circulant matrix can be expressed in a more explicit form to better understand its structure. To this end, we fix the notations \(x_1 = \theta\), \(x_2 = \delta\), and $a_{0,t}=h_t$ for $0\leq t\leq m-1.$  Let \(\eta_{s_{x_1}}\) and \(\eta_{s_{x_2}}\) denote the number of occurrences of \(x_1\) and \(x_2\), respectively, in the expression \(x_{i_1} x_{i_2} \cdots x_{i_s}\). With these notations, the matrix in Definition \ref{5d1} takes the following form
\[
\delta_{\theta}\textup{-}\textup{circ}(h_0,h_1,\dots,h_{m-1})= (a_{st}) = \left(\sum_{l=0}^{s} \sum_{\substack{\eta_{s_{x_1}} = l, \\ \eta_{s_{x_2}} = s-l}} x_{i_1} x_{i_2} \dots x_{i_s} (\sigma^l(h_t)) \right),
\]
where \(0 \leq s, t \leq m-1\), \(\{i_1, i_2, \dots, i_s\} = \{1, 2\}\), and \(\sigma = \begin{bmatrix} h_0 & h_{m-1} & h_{m-2} & \dots & h_1 \\ h_{m-1} & h_{m-2} & h_{m-3} & \dots & h_0 \end{bmatrix}\),
is a permutation on the symbols $\{h_0,h_1,\dots,h_{m-1}\}.$

We now consider the following \(\delta_{\theta}\)-circulant MDS matrix:

\begin{example}
	Let \(\mathbb{F}_{2^4} = \mathbb{F}_2(\alpha)\), where \(\alpha^4 + \alpha + 1 = 0\). Let \(\theta: \mathbb{F}_{2^4} \longrightarrow \mathbb{F}_{2^4}\) be the automorphism defined by \(\theta(a) = a^4\) and let the \(\theta\)-derivation \(\delta: \mathbb{F}_{2^4} \longrightarrow \mathbb{F}_{2^4}\) be given by \(\delta(a) = \alpha (\theta(a) - a)\). For \(a_{0,0} = 1\), \(a_{0,1} = \alpha^3\), \(a_{0,2} = \alpha\), and \(a_{0,3} = \alpha^2 + \alpha + 1\), the \(\delta_\theta\)-circulant matrix is
	\[
	B = \begin{bmatrix}
		1 & \alpha^3 & \alpha & \alpha^2 + \alpha + 1 \\
		\alpha^2 + \alpha + 1 & \alpha^3 + \alpha^2 + \alpha + 1 & \alpha^3 + \alpha^2 + 1 & \alpha + 1 \\
		\alpha & \alpha^3 + 1 & \alpha^2 & \alpha^3 + 1 \\
		\alpha^3 + \alpha^2 & \alpha^3 + \alpha^2 + 1 & \alpha^3 + \alpha^2 & \alpha^3 + \alpha + 1
	\end{bmatrix}_{4 \times 4}.
	\]
Moreover, the matrix $B$ is MDS.
\end{example}

\begin{remark}  By applying specific conditions on the automorphism \(\theta\) and the derivation \(\delta\), we obtain the following immediate observations:
	
	\begin{enumerate}  
		\item If we set $\delta = 0$ in Definition \ref{5d1}, then we obtain a $\theta$-circulant matrix 
		$$A=\begin{bmatrix}
			a_{0,0}&a_{0,1}&\cdots&a_{0,m-1}\\
			\theta(a_{0,m-1})&\theta(a_{0,0})&\cdots&\theta(a_{0,m-2})\\
			\vdots&\vdots&\ddots&\vdots\\
			\theta(a_{0,1})&\theta(a_{0,2})&\cdots&\theta(a_{0,0})
			
		\end{bmatrix}_{m\times m}.$$  
		\item If we set $\delta = 0$  and $\theta=I|_{\mathbb{F}_q}$ in Definition \ref{5d1}, then we obtain circulant matrix 
		$$A=\begin{bmatrix}
			a_{0,0}&a_{0,1}&\cdots&a_{0,m-1}\\
			a_{0,m-1}&a_{0,0}&\cdots&a_{0,m-2}\\
			\vdots&\vdots&\ddots&\vdots\\
			a_{0,1}&a_{0,2}&\cdots&a_{0,0}
		\end{bmatrix}_{m\times m}.$$
	\end{enumerate}  
\end{remark}  
\noindent Throughout this section, we fix some notations.  Let $\theta:\mathbb{F}_{q}\longrightarrow \mathbb{F}_{q}$ be an automorphism of order $m$, that is, $\theta^m(a)=a$ for all $a\in \mathbb{F}_q$ and $\delta:\mathbb{F}_{q}\longrightarrow \mathbb{F}_{q}$ be $\theta$-derivation such that $\delta \circ \theta=\theta \circ  \delta.$
For \( X^m - 1 \in \mathbb{F}_q[X; \theta, \delta] \), \( J = ( X^m - 1 )\) be the two-sided ideal generated by \( X^m - 1 \) if \( m = 2^t \). Now, define a quotient ring using the ideal \( J \) as follows:
$$\frac{\mathbb{F}_q[X;\theta,\delta]}{( X^m-1 )}=\Biggl\{b_0+b_1X+\cdots+b_{m-1}X^{m-1}+( X^m-1);~b_i \in \mathbb{F}_q,~0\leq i \leq m-1\Biggr\}.$$


The following theorem establishes a relationship between a polynomial \( h\langle X\rangle = (X^{m} - 1) + \sum_{i=0}^{m - 1} h_i X^i \in \mathbb{F}_{q}[X; \theta, \delta] \) and  \(\delta_\theta\text{-}\mathrm{circ}(h_0, h_1, \dots, h_{m-1})\) matrix.

\begin{theorem}\label{5l1}
	Let $\theta:\mathbb{F}_{q}\longrightarrow \mathbb{F}_{q}$ be an automorphism and $\delta:\mathbb{F}_{q}\longrightarrow \mathbb{F}_{q}$ be a $\theta$-derivation. Next, let $h\langle X \rangle=(X^m-1)+\sum_{i=0}^{m-1}h_iX^i \in \mathbb{F}_{q}[X;\theta,\delta]$ be a monic polynomial of degree $m$. Then, the matrix associated with the linear transformation
	\begin{eqnarray*}
		\psi:\frac{\mathbb{F}_{q}[X;\theta;\delta]}{(X^m-1)}&\longrightarrow& \frac{\mathbb{F}_{q}[X;\theta;\delta]}{(X^m-1)}\\
		\psi\Big(Q\langle X \rangle\Big)&:=& Q\langle X \rangle*h\langle X \rangle,
	\end{eqnarray*}
	with respect to the basis $\{1,X,X^2,\dots,X^{m-1}\}$ is of the form  
	$$A=(a_{su})=\Bigg(\sum_{i=0}^{s}\binom{s}{i}\theta^i \circ \delta^{s-i}(\sigma^i(h_u))\Bigg),$$
	where $~0\leq s,~u\leq m-1$, and $\sigma=\begin{bmatrix}
		h_0&h_{m-1}&h_{m-2}&\dots&h_1\\
		h_{m-1}&h_{m-2}&h_{m-3}&\dots&h_0
	\end{bmatrix}$ be a permutation on set $\{h_0,h_1,\dots,h_{m-1}\}.$
\end{theorem}
\begin{proof}
	Given that $\theta:\mathbb{F}_{q}\longrightarrow \mathbb{F}_{q}$ be an automorphism and $\delta:\mathbb{F}_{q}\longrightarrow \mathbb{F}_{q}$ be $\theta$-derivation. Next, let $h\langle X \rangle=(X^m-1)+\sum_{i=0}^{m-1}h_iX^i \in \mathbb{F}_{q}[X;\theta,\delta]$ be a monic polynomial of degree $m$. To prove that $A$ is the matrix associated with the linear transformation  
	\[
	\psi: \frac{\mathbb{F}_{q}[X; \theta; \delta]}{(X^m - 1)} \longrightarrow \frac{\mathbb{F}_{q}[X; \theta; \delta]}{(X^m - 1)}
	\]  
	defined by  
	\[
	\psi\Big(Q\langle X \rangle\Big) = Q\langle X \rangle * h\langle X \rangle,
	\]  
	with respect to the basis $\{1, X, X^2, \dots, X^{m-1}\}$, it suffices to show 
	\begin{align}\label{5eqn31}
		X^k * h   \langle X \rangle  = 
		& \sum_{j=0}^{m-1} \Bigg(\sum_{i=0}^{k} 
		\theta^i \circ \delta^{k-i} (\sigma^i(h_{j}))\Bigg) X^{j}~ 
	\end{align}
	for all $k \in \{ 1, 2, \dots, m-1\}$. We prove this result by mathematical induction on $k$.  For $k = 1$, we have  
	\begin{eqnarray}
		X * h\langle X \rangle \notag &=& X * (h_0 + h_1X + \cdots + h_{m-1}X^{m-1}) \\  
		\notag &=& (\theta(h_0)X + \delta(h_0)) + (\theta(h_1)X^2 + \delta(h_1)X) + \cdots   
		 + (\theta(h_{m-1})X^m + \delta(h_{m-1})X^{m-1}) \\  
		\notag &=& (\delta(h_0) + \theta(h_{m-1})) + (\delta(h_1) + \theta(h_0))X + \cdots 
		+ (\delta(h_{m-1}) + \theta(h_{m-2}))X^{m-1}\\
		&=&\sum_{j=0}^{m-1} \Bigg(\sum_{i=0}^{1} 
		\theta^i \circ \delta^{1-i} (\sigma^i(h_{j}))\Bigg) X^{j},
	\end{eqnarray}  
	which is the same as (\ref{5eqn31}) when we take $k = 1$.   Let us assume that the results holds for  $1 < k < m - 1$. We now prove the result for $k + 1$.
	\begin{eqnarray*}
		X^{k+1}*h\langle X\rangle&=&X*(X^k*h\langle X\rangle)\\
		&=& X*\Bigg(\sum_{j=0}^{m-1} \Bigg(\sum_{i=0}^{k} 
		\theta^i \circ \delta^{k-i} (\sigma^i(h_{j}))\Bigg) X^{j}\Bigg)  \\
		&=& \sum_{j=0}^{m-1}\theta\Bigg(\sum_{i=0}^{k} 
		\theta^i \circ \delta^{k-i} (\sigma^i(h_{j})) \Bigg)X^{j+1}+\sum_{j=0}^{m}\delta\Bigg(\sum_{i=0}^{k} 
		\theta^i \circ \delta^{k-i} (\sigma^i(h_{j})) \Bigg)X^j\\
		&=& \sum_{j=0}^{m-1} \Bigg(\sum_{i=0}^{k+1} 
		\theta^i \circ \delta^{k+1-i} (\sigma^i(h_{j}))\Bigg) X^{j},
	\end{eqnarray*}

	\noindent where $\sigma$ is a permutation on a set $\{h_0,~h_1,\ldots,h_{m-1}\}$ and is defined as 
	$$\sigma=\begin{bmatrix}
		h_0&h_{m-1}&h_{m-2}&\dots&h_1\\
		h_{m-1}&h_{m-2}&h_{m-3}&\dots&h_0
	\end{bmatrix}.$$
	This completes the proof.
\end{proof}

\begin{corollary}
	Let $C_{h,\theta,\delta}$ and $C_{p,\theta,\delta}$ are the $\delta_{\theta}$-circulant matrices associated with polynomials $h\langle X \rangle=(X^m-1)+\sum_{i=0}^{m-1}h_iX^i, ~p\langle X \rangle=(X^m-1)+\sum_{i=0}^{m-1}p_iX^i \in \mathbb{F}_{q}[X;\theta,\delta]$, respectively. Then, $C_{h,\theta,\delta}\cdot C_{g,\theta,\delta}$ also $\delta_\theta$-circulant matrix.
\end{corollary}
If we take $\delta=0$ in Theorem~\ref{5l1}, we obtain the following corollary:
 
 \begin{corollary}\cite[Proposition 2]{cauchois2019circulant}
	Let $\theta:\mathbb{F}_{q}\longrightarrow \mathbb{F}_{q}$ be an automorphism.  Next, let $h\langle X \rangle=(X^m-1)+\sum_{i=0}^{m-1}h_iX^i \in \mathbb{F}_{q}[X;\theta]$ be a monic polynomial of degree $m$. Then, the matrix associated with the linear transformation
\begin{eqnarray*}
	\psi:\frac{\mathbb{F}_{q}[X;\theta]}{(X^m-1)}&\longrightarrow& \frac{\mathbb{F}_{q}[X;\theta]}{(X^m-1)}\\
	\psi\Big(Q\langle X \rangle\Big)&:=& Q\langle X \rangle*h\langle X \rangle,
\end{eqnarray*}
with respect to the basis $\{1,X,X^2,\dots,X^{m-1}\}$ is of the form  
$$A=(a_{su})=\Bigg(\theta^s (\sigma^s(h_u))\Bigg),$$
where $~0\leq s,~u\leq m-1$, and $\sigma=\begin{bmatrix}
	h_0&h_{m-1}&h_{m-2}&\dots&h_1\\
	h_{m-1}&h_{m-2}&h_{m-3}&\dots&h_0
\end{bmatrix}$ be a permutation on set $\{h_0,h_1,\dots,h_{m-1}\}.$ 	
 \end{corollary}
 If we take $\delta=0$ and $\theta$ as an identity map in Theorem~\ref{5l1}, we obtain the following corollary:
 
  \begin{corollary}\cite[Proposition 1]{cauchois2019circulant}
  	Let $\theta:\mathbb{F}_{q}\longrightarrow \mathbb{F}_{q}$ be an automorphism.  Next, let $h\langle X \rangle=(X^m-1)+\sum_{i=0}^{m-1}h_iX^i \in \mathbb{F}_{q}[X]$ be a monic polynomial of degree $m$. Then, the matrix associated with the linear transformation
  	\begin{eqnarray*}
  		\psi:\frac{\mathbb{F}_{q}[X]}{(X^m-1)}&\longrightarrow& \frac{\mathbb{F}_{q}[X]}{(X^m-1)}\\
  		\psi\Big(Q\langle X \rangle\Big)&:=& Q\langle X \rangle*h\langle X \rangle,
  	\end{eqnarray*}
  	with respect to the basis $\{1,X,X^2,\dots,X^{m-1}\}$ is of the form  
  	$$A=(a_{su})=(\sigma^s(h_u)),$$
  	where $~0\leq s,~u\leq m-1$, and $\sigma=\begin{bmatrix}
  		h_0&h_{m-1}&h_{m-2}&\dots&h_1\\
  		h_{m-1}&h_{m-2}&h_{m-3}&\dots&h_0
  	\end{bmatrix}$ be a permutation on set $\{h_0,h_1,\dots,h_{m-1}\}.$ 	
  \end{corollary}

	\begin{definition}\cite[Definition 6.3]{MacWilliams1977the}
		For a given  polynomial  
		$
		f\langle X \rangle  = a_0 + a_1X + \dots + a_nX^n$  
		over a finite field \( \mathbb{F}_q \), the weight of \( f\langle X \rangle \) is  define as
		$$
		wt(f\langle X \rangle ) = |\{ i \mid a_i \neq 0, \ 0 \leq i \leq n \}|.$$
	\end{definition}
	

	Based on the $\delta_\theta$-circulant matrix constructed in Theorem \ref{5l1}, we are able to give an algebraic necessary and sufficient condition for  such a matrix to be MDS.
	\begin{theorem}\label{wtmds}
		Let $h\langle X \rangle=(X^m-1)+\sum_{i=0}^{m-1}h_iX^i\in \mathbb{F}_{q}[X;\theta,\delta]$. Then, $\delta_\theta\)-circ\((h_0, h_1, \dots, h_{m-1})$ is MDS if and only if for all $Q_1\langle X \rangle\in \mathbb{F}_{q}[X;\theta,\delta]$, we have,
		$$wt(Q_1\langle X \rangle)+wt(Q_1\langle X \rangle h\langle X \rangle\ \mod  (X^m-1))\geq m+1.$$
	\end{theorem}
	\begin{proof}
		Let $A=\delta_\theta$-circ$(h_0,h_1,\dots,h_{m-1})$ is MDS.
		\begin{eqnarray*}
			&\Leftrightarrow& (I_m|A)~\textup{is the generator matrix of MDS code}\\
			&\Leftrightarrow& \textup{for all } (q_0,q_1,\dots,q_{m-1})\in \mathbb{F}^m_{q},~wt((q_0,q_1,\dots,q_{m-1})\cdot(I_m|A))\geq m+1\\
			&\Leftrightarrow& wt(q_0,q_1,\dots,q_{m-1})+wt((q_0,q_1,\dots,q_{m-1})\cdot A)\geq m+1.
		\end{eqnarray*}
		If one consider $Q_1\langle X \rangle=\sum_{i=0}^{m-1}q_iX^i$, then 
		$$wt(q_0,q_1,\dots,q_{m-1})=wt(Q_1).$$
		From Theorem \ref{5l1}, we know that \( A \) corresponds to right multiplication by \( h \langle X \rangle \) in \( \frac{\mathbb{F}_{q}[X; \theta, \delta]}{(X^m - 1)} \). Thus, we have:
		$$wt((q_0,q_1,\dots,q_{m-1})\cdot A)=wt(Q_1\langle X \rangle h\langle X \rangle\ \mod  (X^m-1)),$$
		which proves the theorem.		
	\end{proof}

\begin{remark}
	It follows from Theorem~\ref{5l1} that the resulting \(\delta_{\theta}\)-circulant MDS matrix has order \(m=2^t\) for some positive integer $t$.
	Note that \(\delta_{\theta}\)-circulant MDS matrices of odd order do exist; however, they do not necessarily arise from the linear transformation described in Theorem~\ref{5l1}.
\end{remark}

 The following example illustrates above remark:

	\begin{example}  
		Let \(\mathbb{F}_{q^n}\) (\(n > 2\)) be a finite field, and the maps \(\theta, \delta: \mathbb{F}_{q^n} \to \mathbb{F}_{q^n}\) defined as  \(\theta(a) = a^2\) and \(\delta(a) = \beta(\theta(a) - a)\), respectively, where \(\beta \in \mathbb{F}_q\). Then, $\delta_{\theta}$-circulant matrix $A$ is defined as  
		\begin{eqnarray*}
			A&=&\delta_\theta\textup{-circ}(\alpha, 1, 1)\\
			&=&\begin{bmatrix}
				\alpha&1&1\\
				\delta(\alpha)+\theta(1)&\delta(1)+\theta(\alpha)&\delta(1)+\theta(1)\\
				\delta^2(\alpha)+\theta^2(1)&\delta^2(1)+\theta^2(\alpha)&\delta^2(1)+\theta^2(1)
			\end{bmatrix}_{3\times 3},
		\end{eqnarray*}  
		where \(\alpha\) is a root of the generating polynomial of \(\mathbb{F}_{2^n}\). This can be easily verified that $A$ is an MDS matrix,  for \(n > 4\).
	\end{example}  
	
	The following theorem outlines the conditions on  polynomial, under which a $\delta_\theta$-circulant matrix is involutory:
	\begin{theorem}  
		Let \( h\langle X \rangle = (X^{m} - 1) + \sum_{i=0}^{m - 1} h_i X^i \in \mathbb{F}_{q}[X; \theta, \delta] \) and \( g\langle X \rangle = (X^{m} - 1) + \sum_{i=0}^{m - 1} g_i X^i \in \mathbb{F}_{q}[X; \theta, \delta] \) be monic polynomials of degree \(m\). Let \( C_{h,\theta,\delta} \) and \( C_{g,\theta,\delta} \) be the \(\delta_\theta\)-circulant matrices associated with polynomials \( h \) and \( g \), respectively. Then,  
		\[
		C_{g,\theta,\delta} \cdot C_{h,\theta,\delta} = I_{m} \quad \text{if and only if} \quad 1 \equiv g\langle X \rangle * h \langle X \rangle \ \mod _*(X^{m} - 1).
		\]
	\end{theorem}
	\begin{proof}  
		Let $$\phi_h, \phi_g: \frac{\mathbb{F}_{q}[X; \theta, \delta]}{(X^{m} - 1)}   \longrightarrow \frac{\mathbb{F}_{q}[X; \theta, \delta]}{(X^{m} - 1)},$$  be the maps. For  $U \langle X \rangle \in  \frac{\mathbb{F}_{q}[X; \theta, \delta]}{(X^{m} - 1)} $, we have
		\[
		\phi_h(U \langle X \rangle) = U \langle X \rangle * h\langle X \rangle \quad \text{and} \quad \phi_g(U \langle X \rangle) = U \langle X \rangle * g\langle X \rangle.
		\]  
		Following Theorem \ref{5l1}, we have the matrices \( C_{h,\theta,\delta} \) and \( C_{g,\theta,\delta} \) associated with the maps \(\phi_h\) and \(\phi_g\), respectively.
		To demonstrate the relationship between the maps \(\phi_h\) and \(\phi_g\), we first consider the composition map  
		\[
		\phi_h \circ \phi_g : \frac{\mathbb{F}_{q}[X; \theta, \delta]}{(X^{m} - 1)}  \longrightarrow \frac{\mathbb{F}_{q}[X; \theta, \delta]}{(X^{m} - 1)} ,
		\]  
		defined as  
		\[
		\phi_h \circ \phi_g(U \langle X \rangle) \longrightarrow U \langle X \rangle * g \langle X \rangle * h \langle X \rangle.
		\]  
		
		\noindent Now, assume that \( g\langle X \rangle * h\langle X \rangle = 1 \ \mod _* (X^{m} - 1) \). Then, for any polynomial \( P\langle X \rangle \in \frac{\mathbb{F}_{q}[X; \theta, \delta]}{(X^{m} - 1)} \), we have
		\begin{eqnarray*}
			\phi_h \circ \phi_g(P \langle X \rangle) &=& \phi_h(P \langle X \rangle * g \langle X \rangle)\\& =& P \langle X \rangle * g \langle X \rangle * h \langle X \rangle \ \mod _* (X^{m} - 1)\\&=& P \langle X \rangle \ \mod _*(X^{m} - 1).
		\end{eqnarray*}

		\noindent Next, we compute the image of the basis set \( B =\{1,~X,~X^2,\dots,~X^{m-1}\}\) under the composition map \(\phi_h \circ \phi_g\)
		\[
		\phi_h\circ\phi_g(1) = 1 \ \mod _* (X^{m} - 1),  
		\]  
		\[
		\phi_h \circ \phi_g(X) = X \ \mod _* (X^{m} - 1),  
		\]  
		\[
		\vdots  
		\]  
		\[
		\phi_h \circ \phi_g(X^{m-1}) = X^{m-1} \ \mod _* (X^{m} - 1).
		\]  
		
		\noindent This establishes that \( [\phi_h \circ \phi_g]_B = I_{m} \), or equivalently,  
		\[
		C_{g,\theta,\delta} \cdot C_{h,\theta,\delta} = I_{m}.
		\]  
		
		\noindent Conversely, suppose that \( C_{g,\theta,\delta} \cdot C_{h,\theta,\delta} = I_{m} \). In view of \cite[Page 67]{Friedberg2013linear}, we have  
		\[
		C_{g,\theta,\delta} \cdot C_{h,\theta,\delta} = [\phi_h \circ \phi_g]_B = I_{m}.
		\]  
		This implies that \(\phi_h \circ \phi_g(X^i) = X^i\) for all \(i = 0, 1, \dots, m - 1\). Thus, we conclude that  
		\[
		\phi_h \circ \phi_g = I|_{\frac{\mathbb{F}_{q}[X; \theta, \delta]}{(X^{m} - 1)}}, \textup{ where } I|_{\frac{\mathbb{F}_{q}[X; \theta, \delta]}{(X^{m} - 1)}} \textup{ is an identity map.}
		\]  
		
		\noindent For any polynomial \( P \langle X \rangle \in \frac{\mathbb{F}_{q}[X; \theta, \delta]}{(X^{m} - 1)}  \), we have  
		\[
		\phi_h \circ \phi_g(P \langle X \rangle) = P \langle X \rangle * g \langle X \rangle * h \langle X \rangle \ \mod  (X^{m} - 1) = P \langle X \rangle \ \mod _* (X^{m} - 1).
		\]  
		This shows  that \( g \langle X \rangle * h \langle X \rangle = 1 \ \mod _* (X^{m} - 1) \).  
	\end{proof}

We conclude this section with a brief summary of our results and by posing an open problem. 
It is well known that every $\theta$-derivation on $\mathbb{F}_q$ admits a representation of the form
\(
\delta(a)=\beta\big(\theta(a)-a\big) \text{ for all } a\in\mathbb{F}_q,
\)
for some $\beta\in\mathbb{F}_q$. 
In this section, we investigated the conditions under which a $\delta_{\theta}$-circulant matrix satisfies the MDS property. 
We obtained a partial characterization in the case where $\beta\in (\mathbb{F}_q)^\theta$. More precisely, under this assumption, Theorem~\ref{5l1} yields an explicit form  of the 
$\delta_{\theta}$-circulant matrix associated with the corresponding linear transformation, namely,
\[
\small{
	\begin{bmatrix}
		h_{0} & h_{1} & \cdots & h_{m-1} \\[1mm]
		\sum_{i=0}^{1} \binom{1}{i}\,\delta^{\,1-i}\theta^i\big(\sigma^i(h_0)\big) &
		\sum_{i=0}^{1} \binom{1}{i}\,\delta^{\,1-i}\theta^i\big(\sigma^i(h_1)\big) &
		\cdots &
		\sum_{i=0}^{1} \binom{1}{i}\,\delta^{\,1-i}\theta^i\big(\sigma^i(h_{m-1})\big) \\
		\vdots & \vdots & \ddots & \vdots \\
		\sum_{i=0}^{m-1} \binom{m-1}{i}\,\delta^{\,m-1-i}\theta^i\big(\sigma^i(h_0)\big) &
		\sum_{i=0}^{m-1} \binom{m-1}{i}\,\delta^{\,m-1-i}\theta^i\big(\sigma^i(h_1)\big) &
		\cdots &
		\sum_{i=0}^{m-1} \binom{m-1}{i}\,\delta^{\,m-1-i}\theta^i\big(\sigma^i(h_{m-1})\big)
	\end{bmatrix}_{m\times m}.
}
\]
Moreover, in Theorem~\ref{wtmds} we derived necessary conditions for such matrices to possess the MDS property. 
However, when $\beta\notin (\mathbb{F}_q)^\theta$, the above form  $\delta_{\theta}$-circulant matrix is no longer available, 
and a complete characterization of MDS $\delta_{\theta}$-circulant matrices remains unknown.  This motivates the following open problem:
\begin{prob}\label{op:dtheta}
		Let $h\langle X \rangle=(X^m-1)+\sum_{i=0}^{m-1}h_iX^i\in \mathbb{F}_{q}[X;\theta,\delta]$. Then, $\delta_\theta\)-circ\((h_0, h_1, \dots, h_{m-1})$ is MDS if and only if for all $Q_1\langle X \rangle\in \mathbb{F}_{q}[X;\theta,\delta]$, we have,
	$$wt(Q_1\langle X \rangle)+wt(Q_1\langle X \rangle h\langle X \rangle\ \mod  (X^m-1))\geq m+1.$$
\end{prob}

\section{\textbf{Classification and equivalence of quasi recursive MDS matrices}}
In this section, we present some fundamental results on the similarity and equivalence of recursive MDS matrices. To establish these results, first we  fix some notions that will be used throughout this section. Let \(\theta: \mathbb{F}_{q} \longrightarrow \mathbb{F}_{q}\) be an automorphism and \(\delta: \mathbb{F}_{q} \longrightarrow \mathbb{F}_{q}\) be a \(\theta\)-derivation such that \(\delta \circ \theta = \theta \circ \delta\). 

Now, we start this section with the following definitions:
\begin{definition}
	Two matrices \( M \) and \( M' \) in \( M(m,\mathbb{F}_q) \) are called diagonally similar if there exists a  diagonal matrix $D\in$ \( GL(m,\mathbb{F}_{q}) \) such that \( M' = D^{-1} M D \).
\end{definition}
\begin{definition}
	Two matrices \( M \) and \( M' \) in \( M(m,\mathbb{F}_q) \) are called permutation similar if there exist permutation matrix \( P \in P(m,\mathbb{F}_2) \) such that  
	\( M' = P M P^{-1} \).
\end{definition}
To define a quasi recursive MDS matrices, we define some notations. For an element $a\in \mathbb{F}_{q}$, we use the notation
\begin{eqnarray}\label{5eqn155}
	a^{\langle i \rangle}=\sum_{k=0}^{i}\binom{i}{k}\delta^{i-k}\theta^{k}(a).
\end{eqnarray}

\noindent We introduce the notation \( a^{\langle i \rangle} \), which plays a crucial role in our analysis. For any two positive integers \(s,t\) such that \(1 \leq s,t \leq m\), let \(A=(a_{s,t})\) be a matrix of order \(m\). We define the matrix \(A^{\langle i \rangle}\) by
\begin{eqnarray}\label{5eqn156}
	A^{\langle i \rangle} &=& \bigl(a^{\langle i \rangle}_{s,t}\bigr).
\end{eqnarray}

If we take $\delta=0$, then  (\ref{5eqn155}) and (\ref{5eqn156})  yields 
\begin{eqnarray}\label{5eqn157}
	a^{[i]}&=&\theta^i(a), \quad A^{[i]}=(\theta^i(a_{s,t})).
\end{eqnarray}
The matrix powers \(A^{\langle i\rangle}\) defined in \eqref{5eqn156} and \eqref{5eqn157} are different from the usual powers \(A^{i}\) obtained via the classical matrix product.

Motivated by the definition of recursive MDS matrix ($r$-MDS matrix), we introduce the definition of quasi recursive MDS matrix (quasi $r$-MDS matrix) involving $\theta$-derivation.
\begin{definition}
	Let $r$ be a positive integer. A matrix $M$ is said to be quasi recursive MDS or quasi $r$- MDS if the matrix $M^{\langle r-1 \rangle }\cdot M^{\langle r-2 \rangle }\cdots M^{\langle 1 \rangle}\cdot M$ is MDS. If $M$ is quasi $r$-MDS, then we say $M$ yields an MDS matrix. Moreover, if $M$ is companion matrix, then we say $M$ is a quasi $r$-MDS companion matrix.
\end{definition}

	Now, we present various auxiliary results, before providing the main results of this section. 
	\begin{proposition}
		Let $A=(a_{s,t})$ and $B=(b_{s,t})$ be two matrices of order $m$ over a finite field $\mathbb{F}_{q}$. Then, 
		$$(A^{\langle i \rangle })^{\langle j \rangle }=A^{\langle i+j \rangle }.$$
		\begin{proof}
			Let $A=(a_{s,t})$ and $B=(b_{s,t})$ be two matrices of order $m$ over  $\mathbb{F}_{q}$. We have to show that $(A^{\langle i \rangle })^{\langle j \rangle }=A^{\langle i+j \rangle }.$  For this, it suffices to show that,   for $a\in \mathbb{F}_{q}$, we have  $(a^{\langle i \rangle })^{\langle j \rangle }=a^{\langle i+j \rangle }.$ Now, we compute
			\begin{eqnarray}\label{5eqn57}
				a^{\langle i+j \rangle}&=& \sum_{k=0}^{i+j}\binom{i+j}{k}\delta^{i+j-k}\theta^{k}(a),
			\end{eqnarray}
			
			\begin{eqnarray}\label{5eqn58}
				a^{\langle i \rangle}&=& \sum_{k=0}^{i}\binom{i}{k}\delta^{i-k}\theta^{k}(a).
			\end{eqnarray}
			Now, for any positive integer $j$,  we compute
			\begin{eqnarray}\label{5eqn59}
				(a^{\langle i \rangle})^{\langle j \rangle }&=& \sum_{k=0}^{j}\binom{j}{k}\delta^k\theta^{j-k}
				\Bigg(\sum_{k'=0}^{i}\binom{i}{{k'}}\delta^{i-k'}\theta^{k'}(a)\Bigg)\nonumber\\
				&=&\sum_{k=0}^{i+j}\Bigg(\sum_{l+t=k}^{}\binom{j}{l}\binom{i}{t}\Bigg)\delta^k\theta^{i+j-k}(a).
			\end{eqnarray}
			Since,
			\begin{eqnarray}\label{5eqn60}
				\sum_{l+t=k}^{}\binom{j}{l}\binom{i}{t}=\binom{i+j}{k},
			\end{eqnarray}
			so from (\ref{5eqn57}), (\ref{5eqn59})  and (\ref{5eqn60}), we have
			$(a^{\langle i \rangle })^{\langle j \rangle }=a^{\langle i+j \rangle }.$ This implies that $(A^{\langle i \rangle })^{\langle j \rangle }=A^{\langle i+j \rangle }.$
		\end{proof}
	\end{proposition}
	
	\begin{example}\label{5example001}
		Let \(\mathbb{F}_{2^4} = \mathbb{F}_2(\alpha)\), where \(\alpha^4 + \alpha + 1 = 0\). Let \(\theta: \mathbb{F}_{2^4} \longrightarrow \mathbb{F}_{2^4}\) be the automorphism defined by \(\theta(a) = a^4\) and let the \(\theta\)-derivation \(\delta: \mathbb{F}_{2^4} \longrightarrow \mathbb{F}_{2^4}\) be given by \(\delta(a) = (\alpha^2+\alpha) (\theta(a) - a)\). Let $A=\begin{bmatrix}
			1&0\\
			\alpha&\alpha
		\end{bmatrix}_{2\times 2}$ and $B=\begin{bmatrix}
			0&1\\
			\alpha&\alpha
		\end{bmatrix}_{2\times 2}$. In this case, we obtain $(AB)^{\langle 1\rangle}\neq A^{\langle 1\rangle}B^{\langle 1\rangle}.$			
	\end{example}
	
Example~\ref{5example001} shows that the product defined in \eqref{5eqn156} does not necessarily distribute over matrix multiplication.
 That is,
\[
(AB)^{\langle i\rangle}\neq A^{\langle i\rangle}B^{\langle i\rangle}.
\]

	\begin{proposition}
		Let $a\in (\mathbb{F}_q)^\theta$. Then,
		$$\delta^k(ab)=a\delta^k(b)~\textup{for all }b \in \mathbb{F}_q.$$
	\end{proposition}
	\begin{proof}  
		The proof can be established using the principle of mathematical induction.  
	\end{proof}

	\begin{proposition}\label{5prop1001}
		Let $a\in (\mathbb{F}_q)^\theta$ and $b\in \mathbb{F}_q.$ Then, 
		$$(ab)^{  \langle i \rangle }=ab^{  \langle i \rangle }.$$
	\end{proposition}
	\begin{proof}
		Given $a\in (\mathbb{F}_q)^\theta$ and $b\in \mathbb{F}_q.$ Then, we compute
		\begin{eqnarray*}
			(ab)^{  \langle i \rangle }&=&\sum_{k=0}^{i}\binom{i}{k}\delta^{i-k}\theta^k(ab)\\
			&=&\sum_{k=0}^{i}\delta^{i-k}(a\theta^k(b))\\
			&=&a\sum_{k=0}^{i}\delta^{i-k}\theta^k(b)\\
			&=&ab^{  \langle i \rangle }.
		\end{eqnarray*}
	\end{proof}
	\begin{proposition}\label{5p001}
		Let $A\in M(m,\mathbb{F}_q)$, and $D$ be any diagonal matrix of order $m$ over $(\mathbb{F}_q)^\theta.$ Then, 
		\begin{enumerate}
			\item[(i)] $(DA)^{  \langle t \rangle }=D^{  \langle t \rangle }A^{  \langle t \rangle }$
			\item[(ii)]$(AD)^{  \langle t \rangle }=A^{  \langle t \rangle }D^{  \langle t \rangle }.$
		\end{enumerate}
		
	\end{proposition}
	\begin{proof}{$(i)$}
		Let $A\in M(m,\mathbb{F}_q)$ and $D\in M(m,(\mathbb{F}_q)^\theta)$ such that
		$$A=\begin{bmatrix}
			a_{0,0}&a_{0,1}&\dots&a_{0,m-1}\\
			a_{1,0}&a_{1,1}&\dots&a_{1,m-1}\\
			\vdots&\vdots&\ddots&\vdots\\
			a_{m-1,0}&a_{m-1,1}&\dots&a_{m-1,m-1}
		\end{bmatrix}_{m\times m}~\textup{ and } D=\begin{bmatrix}
			d_0&0&\dots&0\\
			0&d_1&\dots&0\\
			\vdots&\vdots&\ddots&\vdots\\
			0&0&\dots&d_{m-1}
		\end{bmatrix}_{m\times m}.$$
		Now, we compute 
		$$DA=\begin{bmatrix}
			d_0a_{0,0}&d_0a_{0,1}&\dots&d_0a_{0,m-1}\\
			d_1a_{1,0}&d_1a_{1,1}&\dots&d_1a_{1,m-1}\\
			\vdots&\vdots&\ddots&\vdots\\
			d_{m-1}a_{m-1,0}&d_{m-1}a_{m-1,1}&\dots&d_{m-1}a_{m-1,m-1}\\
			
		\end{bmatrix}_{m\times m}$$
		\begin{eqnarray*}
			(DA)^{  \langle t \rangle }&=&((d_ia_{ij})^{  \langle t \rangle }).
		\end{eqnarray*}
		In view of Proposition \ref{5prop1001}, we obtain
		\begin{eqnarray*}
			(DA)^{  \langle t \rangle }&=&((d_ia_{ij})^{  \langle t \rangle })\\
			&=&(d^{  \langle t \rangle }_ia^{  \langle t \rangle }_{ij})\\
			&=&(d_ia^{  \langle t \rangle }_{ij})\\
			&=&D^{\langle t \rangle}A^{\langle t \rangle}.
		\end{eqnarray*}
		This implies that 
		$$(DA)^{\langle t \rangle}= D^{\langle t \rangle}A^{\langle t \rangle}.$$
		
		$(ii)$	Similarly, we can prove that $(AD)^{  \langle t \rangle }=A^{  \langle t \rangle }D^{  \langle t \rangle }$.
	\end{proof}
	\begin{corollary}
		Let $M_1\in M(m,\mathbb{F}_q)$ and diagonal matrices $D_1,~D_2\in M(m,(\mathbb{F}_q)^\theta)$. Then,
		$$(D_1M_1D_2)^{  \langle t \rangle }=D^{  \langle t \rangle }_1M^{  \langle t \rangle }_1D^{  \langle t \rangle }_2.$$
	\end{corollary}
	Similarly, we can prove the following proposition: 
	\begin{proposition}\label{5p002}
		Let $A\in M(m,\mathbb{F}_q)$ and $P,~P_1,~P_2\in P(m,\mathbb{F}_2)$. Then,
		\begin{enumerate}
			\item[(i)] $(PA)^{  \langle t \rangle }=P^{  \langle t \rangle }A^{  \langle t \rangle }$
			\item [(ii)]$(AP)^{  \langle t \rangle }=A^{  \langle t \rangle }P^{  \langle t \rangle }$
			\item[(iii)] $(P_1AP_2)^{  \langle t \rangle }=P_1^{  \langle t \rangle }A^{  \langle t \rangle }P_2^{  \langle t \rangle }.$
		\end{enumerate}
	\end{proposition}

	Based on the Proposition \ref{5p001} and \ref{5p002}, we have the following result:
	\begin{theorem}\label{5thm4.10}
		Let $\theta:\mathbb{F}_q\rightarrow \mathbb{F}_q$ be an automorphism and $\delta:\mathbb{F}_q\rightarrow\mathbb{F}_q$ be a derivation such that $\delta \circ \theta= \theta \circ  \delta.$  
		\begin{enumerate}
			\item[(i)] Suppose that two matrices $M_1,~M_2\in M(m,\mathbb{F}_q)$ are diagonally similar such that 
			$$M_2=DM_1D^{-1},~~D\in GL(m,(\mathbb{F}_q)^\theta).$$
			Then, $M_1$ is quasi $r$-MDS iff $M_2$ is a quasi $r$-MDS.
			\item[(ii)] 	 Suppose that two matrices $M_1,~M_2\in M(m,\mathbb{F}_q)$ are permutation  similar such that 
			$$M_2=PM_1P^{-1},~P\in P(m,\mathbb{F}_2)$$
			Then, $M_1$ is quasi $r$-MDS iff $M_2$ is a quasi $r$-MDS.
		\end{enumerate}
	\end{theorem}
	\begin{proof}
		$(i)$ 
		We have $M_1$, $M_2\in M(m,\mathbb{F}_q)$ are diagonally similar such that 
		$$M_2=DM_1D^{-1}.$$
		Let us suppose $M_1$ is quasi $r$-MDS matrix. This implies that, there exists positive integer $r$ such that  $M_1^{\langle r-1 \rangle }\cdot M_1^{\langle r-2 \rangle }\cdots M_1^{\langle 1 \rangle}\cdot M_1$, is a quasi $r$-MDS. Now, we compute
		\begin{eqnarray}\label{5eqn4.9}
			\noindent		M_2^{\langle r-1 \rangle }\cdot M_2^{\langle r-2 \rangle }\cdots M_2^{\langle 1 \rangle}\cdot M_2\notag&=&(DM_1D^{-1})^{\langle r-1\rangle}(DM_1D^{-1})^{\langle r-2\rangle }\cdots (DM_2D^{-1}).\\
			&~&
		\end{eqnarray}
		Since $D\in GL(m,(\mathbb{F}_q)^\theta)$, so  Proposition \ref{5p001} and (\ref{5eqn4.9}) yields
		\begin{eqnarray*}
			M_2^{\langle r-1 \rangle }\cdot M_2^{\langle r-2 \rangle }\cdots M_2^{\langle 1 \rangle}\cdot M_2&=&D	M_1^{\langle r-1 \rangle }\cdot M_1^{\langle r-2 \rangle }\cdots M_1^{\langle 1 \rangle}\cdot M_1D^{-1}.
		\end{eqnarray*}
		This shows $M_2$ is quasi $r$-MDS matrix.\\\\
		$(ii)$
		By using Proposition \ref{5p002}, one may easily prove this result by using the similar arguments that were used to prove the part $(i)$.
	\end{proof}
\section{\textbf{Characterization of polynomials that gives quasi recursive MDS matrices}}

    In the previous section, we introduced the notion of quasi $r$-MDS matrices and highlighted their role in producing MDS matrices. In this section, we focus on the analysis and construction of such matrices arising specifically from companion matrices associated with monic skew polynomials in $\mathbb{F}_q[X;\theta]$.   Following the framework of Cauchois et al.~\cite{cauchois2016direct}, we recall that a quasi recursive MDS matrix admits a factorization of the form
\[
M \;=\; C_g^{[r-1]}\, C_g^{[r-2]}\cdots C_g^{[1]} \, C_g,
\]
where $g(X)\in \mathbb{F}_q[X;\theta]$ is a monic polynomial of degree $m$, and $r\ge m$. In this case, we say that $g(X)$ yields a quasi recursive MDS matrix.

 Cauchois et al.~\cite{cauchois2016direct} established constructions satisfying the quasi-involutory identity $M\cdot M^{[m]}=I_m$. In contrast, our method yields quasi recursive MDS matrices that are  involutory. For this purpose, given any square matrix $A=(a_{s,t})$ of order $m$, and by taking $\delta=0$ in (\ref{5eqn156}), we adopt the notation
 \begin{eqnarray}\label{5eqn56}
 	A^{[i]}=(a^{[i]}_{s,t}), \qquad 1\le s,t\le m.
 \end{eqnarray}
 
We now recall the notion of $\theta$-cyclic codes over a finite field, which plays a central role in our development. A linear code $C\subseteq \mathbb{F}_q^n$ is called $\theta$-cyclic (or skew cyclic) if it is closed under the $\theta$-shift, that is,
\[
(c_0,c_1,\dots,c_{n-1})\in C
\ \Longrightarrow\
(\theta(c_{n-1}),\theta(c_0),\dots,\theta(c_{n-2}))\in C.
\]
When $|\theta|\mid n$, the polynomial $X^n-1$ lies in the center of the skew polynomial ring $\mathbb{F}_q[X;\theta]$, and therefore $\langle X^n-1\rangle$ becomes a two-sided ideal. This allows us to consider the quotient ring
\[
\frac{\mathbb{F}_q[X;\theta]}{\langle X^n-1\rangle}
=\left\{a_0+a_1X+\cdots+a_{n-1}X^{n-1}+\langle X^n-1\rangle \; ;\; a_i\in \mathbb{F}_q\right\}.
\]

Throughout this paper, we assume that $|\theta|$ divides $n$.  Consequent to the above discussion, a $\theta$-cyclic code admits two equivalent descriptions as follows:
\begin{itemize}
	\item[\textnormal{(i)}] It is an $\mathbb{F}_{q}$-submodule of $\mathbb{F}_{q}^{n}$ such that
	\[
	C:=\left\{(c_0,c_1,\dots,c_{n-1})\in \mathbb{F}_{q}^{n};\;
	(\theta(c_{n-1}),\theta(c_0),\dots,\theta(c_{n-2}))\in C\right\};
	\]
	\item[\textnormal{(ii)}] It is a left ideal in the quotient ring $\frac{\mathbb{F}_{q}[X;\theta]}{\langle X^n-1\rangle}$, i.e.,
	\[
	C:={ }^\ast\langle g(X)\rangle
	=\left\{\,a(X)g(X)\;\mid\; a(X)\in \frac{\mathbb{F}_{q}[X;\theta]}{\langle X^n-1\rangle}\right\}.
	\]
\end{itemize}

To this end, let
\(
g(X)=g_0+g_1X+\cdots+g_{m-1}X^{m-1}+X^m\in \mathbb{F}_{q}[X;\theta]
\)
be a right divisor of $X^{n}-1$, and let $C$ be the principal $\theta$-cyclic code, i.e.,
\(
C={ }^\ast\langle g(X)\rangle.
\)
A generator matrix of $C$ is given by
\begin{eqnarray}\label{Equation_gen0_Fq}
	G=\begin{bmatrix}
		g(X)\\
		X\ast g(X)\\
		\vdots\\
		X^{n-m-1}\ast g(X)
	\end{bmatrix}_{n-m\times n}.
\end{eqnarray}

In systematic form, this generator matrix can be rewritten as
\begin{eqnarray}\label{Equation_gen_Fq}
	\begin{bmatrix}
		- X^m\bmod_\ast g & 1 & 0 & \cdots & 0 \\
		- X^{m+1} \bmod_\ast g & 0 & 1 & \cdots & 0 \\
		\vdots & \vdots & \vdots & \ddots & \vdots \\
		- X^{m-1} \bmod_\ast g & 0 & 0 & \cdots & 1
	\end{bmatrix}_{n-m\times n}.
\end{eqnarray}

In particular, if we take $n=2m$, then (\ref{Equation_gen_Fq}) reduces to
\[
\begin{bmatrix}
	- X^m \bmod_\ast g & 1 & 0 & \cdots & 0 \\
	- X^{m+1} \bmod_\ast g & 0 & 1 & \cdots & 0 \\
	\vdots & \vdots & \vdots & \ddots & \vdots \\
	- X^{2m-1} \bmod_\ast g & 0 & 0 & \cdots & 1
\end{bmatrix}_{m\times 2m}
=
\begin{bmatrix}
	g_0&g_1&g_2&\dots&g_{m-1}&1&0&\dots&0\\
	0&g^{[1]}_0&g^{[1]}_1&\dots&g^{[1]}_{m-2}&g^{[1]}_{m-1}&1&\dots&0\\
	\vdots&\vdots&\vdots&\ddots&\vdots&\vdots\\
	0&0&0&\dots&g^{[m-1]}_0&g^{[m-1]}_1&\dots&g^{[m-1]}_{m-1}&1\\
\end{bmatrix}_{m\times 2m}.
\]

Now, isolating the redundant part, we obtain the matrix
\begin{eqnarray}\label{5eqn53}
	N_g=\begin{bmatrix}
		X^m \bmod_\ast g(x) \\
		X^{m+1} \bmod_\ast g(x) \\
		\vdots \\
		X^{2m-1} \bmod_\ast g(x)
	\end{bmatrix}_{m\times m}.
\end{eqnarray}

\noindent To analyze $N_g$, it is convenient to introduce the companion matrix associated with the polynomial $g(X)$.

\begin{definition}
	Let
	\(
	g(X)=g_0+g_1X+g_2X^2+\cdots+g_{m-1}X^{m-1}+X^m\in \mathbb{F}_{q}[X;\theta]
	\)
	be a monic polynomial of degree $m$ over the finite field $\mathbb{F}_{q}$.
	The companion matrix of $g(X)$, denoted by $C_g$, is defined as
	\[
	C_g=
	\begin{bmatrix}
		0 & 1 & 0 & \cdots  & 0 \\
		0 & 0 & 1 & \cdots  & 0 \\
		\vdots & \vdots & \vdots & \ddots & \vdots \\
		0 & 0 & 0 & \cdots  & 1 \\
		-g_0 & -g_1 & -g_2 & \cdots  &- g_{m-1}
	\end{bmatrix}_{m\times m}.
	\]
\end{definition}

For each integer $i\ge 0$, we denote by $C_g^{[i]}$ the matrix obtained by applying $\theta^i$ to every entry of $C_g$. Explicitly,
\[
C_g^{[i]}=
\begin{bmatrix}
	0 & 1 & 0 & \cdots  & 0 \\
	0 & 0 & 1 & \cdots & 0 \\
	\vdots & \vdots & \vdots & \ddots &  \vdots \\
	0 & 0 & 0 & \cdots &  1 \\
	-\theta^i(g_0) &- \theta^i(g_1) &- \theta^i(g_2) & \cdots &- \theta^i(g_{m-1})
\end{bmatrix}_{m\times m}.
\]

\begin{eqnarray}
	(C_{g^{[i]}})^{-1}=\begin{bmatrix}
		\frac{\theta^{i}(g_1)}{\theta^{i}(g_0)}&\frac{\theta^{i}(g_2)}{\theta^{i}(g_0)}&\cdots&\frac{\theta^{i}(g_{m-1
			})}{\theta^{i}(g_0)}&\frac{1}{\theta^{i}(g_0)}\\
		1&0&\cdots&0&0\\
		\vdots&\vdots&\ddots&\vdots&\vdots\\
		0&0&\cdots&1&0\\
		
	\end{bmatrix}_{m\times m}.
\end{eqnarray}

Note that, if $A$ is a square matrix of order $k$, say $A=(a_{ij})$ with $a_{ij}\in \mathbb{F}_{q}$ for $1\le i,j\le m$, then
\(
A^{[t]}=(\theta^t(a_{ij}))
\)
for some $t\in \mathbb{N}\cup \{0\}$.


In 2016, Cauchois et al. \cite{cauchois2016direct}, proved some multiplication rules for the matrix of the form (\ref{5eqn56}) as follows:
\begin{proposition}\cite[Proposition 1]{cauchois2016direct}\label{5prop13}
	Let $A=(a_{s,t})$ and $B=(b_{s,t})$ be two matrices of order $m$ over a finite field $\mathbb{F}_{q}$. Then, we have
	\begin{enumerate}
		\item[(i)] $(A^{[i]})^{[j]}=A^{[i+j]}.$
		\item[(ii)] $(AB)^{[i]}=A^{[i]}B^{[i]}.$
	\end{enumerate}
\end{proposition}

Following the approach of Cauchois et al.~\cite[Theorem~2]{cauchois2016direct}, the next lemma represents the redundant matrix
\(N_g\) in \eqref{5eqn53} as a structured product of successive
\(\theta\)-twists of the companion matrix. We are now in a position to establish one of the main results of this section, namely, the construction of involutory MDS matrices via the quasi recursive method.

\begin{theorem}\label{5t1}
	Let $g  \langle X \rangle=g_0+g_1X+\cdots+g_{m-1}X^{m-1}+X^m\in\mathbb{F}_{q^{2m}}[X;\theta]$ be a polynomial of degree $m$ and an automorphism $\theta:\mathbb{F}_{q^{2m}}\longrightarrow \mathbb{F}_{q^{2m}}$, defined as $\theta(a)=a^{q^2}.$ Then, $N_g$ satisfies the quasi recursive property, i.e.,
	$$	N_g=C_g^{[m-1]}\cdot C_g^{[m-2]}\cdots C_g^{[1]}\cdot C_g.$$
	Moreover, if $g$ is the generator polynomial of a  MDS code, then $N_g$ is MDS and satisfies the involutory property, i.e.,
	$	N^{2}_g=I_m.$
\end{theorem}
\begin{proof}
	We show by induction that for all integers $i\geq 1$, the following property $P_i$ are satisfied:
	\begin{eqnarray}\label{5eqn201}
		P_i&:&\begin{bmatrix}
			X^i ~ \mod _*g \\
			X^{i+1} ~ \mod _*g  \\
			\vdots\\
			X^{i+m-1} ~ \mod _*g  \\
		\end{bmatrix}_{m\times m}=C^{[i-1]}_g\cdot C^{[i-2]}_g\cdots C^{[1]}_g\cdot C_g.
	\end{eqnarray}
	For $i=0$ 
	\begin{eqnarray}
		P_0&:&\begin{bmatrix}
			X^0 ~ \mod _*g  \\
			X^{1} ~ \mod _*g  \\
			\vdots\\
			X^{m-1} ~ \mod _*g  \\
			
		\end{bmatrix}_{m\times m}=C_g.
	\end{eqnarray}
	Thus (\ref{5eqn201}) is true for $i=0$. Suppose our assumption is true for $1,~2,~\dots,~i-1$. Then, we have
	
	
	$$\begin{bmatrix}\label{5eqn5.8}
		X^i ~ \mod _*g  \\
		X^{i+1} ~ \mod _*g  \\
		\vdots\\
		X^{i+m-1} ~ \mod _*g  \\
	\end{bmatrix}_{m\times m}=
	C_g^{[i-1]}\cdot C_g^{[i-2]}\cdots C_g^{[1]}\cdot C_g.$$
	Now, we prove (\ref{5eqn201}) for $i$. For this, we compute \\
	\begin{eqnarray}\label{5eqn11}
		\notag&~&C^{[i]}_g C_g^{[i-1]}\cdot C_g^{[i-2]}\cdots C_g^{[1]}\cdot C_g\\
		\notag&~~~~~&	=\begin{bmatrix}
			0&1&0&\cdots&0\\
			0&0&1&\cdots&0\\
			\vdots&\vdots&\vdots&\ddots&\vdots&\\
			0&0&0&\cdots&1\\
			\theta^{i}(g_0)&\theta^{i}(g_1)&\theta^{i}(g_2)&\cdots&\theta^{i}(g_{m-1})
		\end{bmatrix}_{m\times m} \cdot \begin{bmatrix}
			X^i ~ \mod _*g  \\
			X^{i+1} ~ \mod _*g \\
			\vdots\\
			X^{i+m-1} ~ \mod _*g \\
		\end{bmatrix}_{m\times m}\\
		\notag&~&~~~~~=\begin{bmatrix}
			X^{i+1}~ \mod _*g  \\
			X^{i+2}~ \mod _*g  \\
			\vdots\\
			\theta^{i}(g_0)X^i+\cdots+\theta^{i}(g_{m-1})X^{i+m-1}~ \mod _*g  
		\end{bmatrix}_{m\times m}\\
		&~~~~~~~~&=\begin{bmatrix}
			X^{i+1}~ \mod _*g \\
			X^{i+2}~ \mod _*g  \\
			\vdots\\
			\Big(\sum_{j=0}^{m-1}\theta^{i}(g_j)X^{j+i}\Big)~ \mod _*g  \\
		\end{bmatrix}_{m\times m}.
	\end{eqnarray}
	\normalsize
	Since,
	\begin{eqnarray*}
		X^i*g_t&=&\theta^i(g_t)X^i.
	\end{eqnarray*}
	So, we have
	\begin{eqnarray}\label{5eqn12}
		~\notag&X^i&*(g_0+g_1X+g_2X^2+\cdots+g_{m-1}X^{m-1})\\
		\notag&=&\theta^i(g_0)X^i+\theta^i(g_1)X^{i+1}+\cdots+\theta^i(g_{m-1})X^{i+m-1}\\
		&=&\sum_{j=0}^{m-1}\theta^i (g_j)X^{i+j}.
	\end{eqnarray}
	From the last row of the matrix in (\ref{5eqn11}), and (\ref{5eqn12}), we obtain
	\begin{eqnarray}\label{5eqn13}
		\Big(\sum_{j=0}^{m-1}\theta^{i}(g_j)X^{j+i}\Big)~ \mod _*g  &=&X^{i+m}~ \mod _*g.
	\end{eqnarray}
	In view of  (\ref{5eqn13}) in (\ref{5eqn11}), we conclude that
	\begin{eqnarray}\label{5eqn202}
		C_g^{[i]}\cdot 	C_g^{[i-1]}\cdot C_g^{[i-2]}\cdots C_g^{[1]}\cdot C_g=\begin{bmatrix}
			X^{i+1} ~ \mod _*g  \\
			X^{i+2} ~ \mod _*g  \\
			\vdots\\
			X^{i+m} ~ \mod _*g  \\
		\end{bmatrix}_{m\times m}.
	\end{eqnarray}
	Hence, $P_i$ is true for all integer $i\geq0.$ If we take $i=2m-1$ in (\ref{5eqn202}), we obtain
	\begin{eqnarray*}
		C^{[2m-1]}_g\cdots C_g&=&\begin{bmatrix}
			X^{2m} ~ \mod _*g  \\
			X^{2m+1} ~ \mod _*g  \\
			\vdots\\
			X^{3m-1} ~ \mod _*g  \\
		\end{bmatrix}_{m\times m}.
	\end{eqnarray*}
	Since \( \theta^m = I|_{\mathbb{F}_{q^{2m}}} \), it follows that \( C_g^{[m]} = C_g \). Taking into account the Proposition \ref{5prop13}, we compute
	
	\begin{eqnarray}\label{5eqn203}
		C^{[2m-1]}_g\cdots C_g \notag&=&C^{[2m-1]}_g\cdots C^{[m]}_g C^{[m]}_g C^{[m-1]}_g\cdots C_g\\
		\notag&=&(C^{[m-1]}_g\cdots C_g)(C^{[m-1]}_g\cdots C_g)\\
		&=&N_g\cdot N_g.
	\end{eqnarray}
	From (\ref{5eqn53}) and (\ref{5eqn203}), we have 
	$$N^2_g=\begin{bmatrix}
		X^{2m} ~ \mod _*g \\
		X^{2m+1} ~ \mod _*g  \\
		\vdots\\
		X^{3m-1} ~ \mod _*g  \\
	\end{bmatrix}_{m\times m}.$$
	As \( X^{2m} = 1 \) in \( \frac{\mathbb{F}_{q^{2m}}[X;~\theta]}{(X^{2m}-1)} \), equation (\ref{5eqn203}) leads to
	\begin{eqnarray*}
		N^2_g=\begin{bmatrix}
			X^{0} ~ \mod _*g  \\
			X^{1} ~ \mod _*g  \\
			\vdots\\
			X^{m-1} ~ \mod _*g  \\
		\end{bmatrix}_{m\times m}.
	\end{eqnarray*}
	This implies that $N^2_g=I_m.$ Since $N_g$ generates an MDS code, and  is the redundant part of systematic form of an MDS code. Hence, $N_g$ is an MDS matrix.
\end{proof}
\begin{example}
	Let $\mathbb{F}_{2^4} = \mathbb{F}_2(\alpha)$, where $\alpha^4 + \alpha + 1 = 0$, and let $\theta : \mathbb{F}_{2^4} \longrightarrow \mathbb{F}_{2^4}$ be a map defined as $\theta(a) = a^4$. In light of (\ref{5eqn111}), we factorize the polynomial  
	\[
	X^4 + 1 = (X^2 + 1) * (X^2 + 1) = (X^2 + \alpha X + \alpha^3 + \alpha)*  (X^2 + \alpha X + \alpha^3 + \alpha^2).
	\]  
	Take \( g(X) = X^2 + \alpha X + \alpha^3 + \alpha^2 \). Then,  
	\[
	N_g = C^{[1]}_g \cdot C_g = 
	\begin{bmatrix}
		0 & 1 \\
		\alpha^3 + \alpha & \alpha^3+\alpha
	\end{bmatrix}_{2\times 2}
	\cdot
	\begin{bmatrix}
		0 & 1 \\
		\alpha^3 + \alpha^2 & \alpha
	\end{bmatrix}_{2\times 2}
	=
	\begin{bmatrix}
		\alpha^3 + \alpha^2 & \alpha \\
		1 + \alpha + \alpha^2 & \alpha^3 + \alpha^2
	\end{bmatrix}_{2\times 2}.
	\]  
	This matrix is an involutory MDS matrix.
\end{example}

Gupta et al. \cite[Theorem 1]{gupta2017towards} provide a key criterion for determining whether a polynomial generates a recursive MDS matrix. We extend this characterization to quasi recursive MDS matrices, generalizing the recursive MDS framework and offering new insights into diffusion matrix construction.
\begin{theorem}\label{5t201}
	Let $g  \langle X \rangle\in \mathbb{F}_{q^{2m}}[X,\theta]$ be a monic polynomial of degree $m$ with $g(0)\neq0$  $ord(g)=n\geq 2$ and $\theta:\mathbb{F}_{q^{2m}}\to \mathbb{F}_{q^{2m}}$, defined as $\theta(a)=a^{q}$. Let $E=\{0,1,\dots,m-1,t,t+1,\dots,t+m-1\}$ for some integer $t$, $m\leq t\leq n-m.$ Then, the matrix 
	$M=C^{[t-1]}_gC^{[t-2]}_g\cdots C^{[1]}C_g$ is MDS iff the weight of any nonzero right multiple  $f  \langle X \rangle$ of $g  \langle X \rangle$ of the form 
	$$f  \langle X \rangle=\sum_{e\in E}^{}f_eX^e\in \mathbb{F}_{q^{2m}}[X,\theta]$$ is greater than $m$.
\end{theorem}
\begin{proof}
	In view of (\ref{5eqn201}), we have 
	$$N_g=C^{[t-1]}_g\cdot C^{[t-2]}_g\cdots C^{[1]}\cdot C_g=\begin{bmatrix}
		X^t~ \mod _*g  \\
		X^{t+1}~\ \mod _*g  \\
		\vdots\\
		X^{t+m-1}~\ \mod _*g  
	\end{bmatrix}_{m\times m}.$$
	The matrix $N_g$ is MDS if and only if any $m$ rows of matrix $\bar{G}$ are linearly independent, where
	$$\bar{G}=\Big[\frac{I}{N_g}\Big]=\begin{bmatrix}
		1\\X\\X^2\\
		\dots\\
		X^{m-1}\\
		X^t~ \mod _*g  \\
		X^{t+1} \mod _*g  \\
		\dots\\
		X^{t+m-1} \mod _*g  
	\end{bmatrix}_{2m\times m}.$$
	Now, it is now easy to see that the latter condition is equivalent to the weight of any nonzero multiple of the form $\geq m$. Hence, this prove the theorem. 
\end{proof}
\begin{definition}
	Let $\mathbb{F}_{q}$ be a finite field and $g  \langle X \rangle=g_0+g_1X+\cdots+g_{m-1}X^{m-1}+X^m$ be a polynomial of degree $m$ over $\mathbb{F}_{2^{2m}}$. The reciprocal polynomial of $g  \langle X \rangle$ denoted by $g_*  \langle X \rangle,$ is defined as:
	$$
	g_*  \langle X \rangle=X^m+	\frac{\theta^{m}(g_1)}{\theta^{m}(g_0)}X^{m-1}+	\frac{\theta^{m}(g_2)}{\theta^{m}(g_0)}X^{m-2}+\cdots+	\frac{\theta^{m}(g_{m-1})}{\theta^{m}(g_0)}X+	\frac{1}{\theta^{m}(g_0)}
	.$$
\end{definition}

\noindent Since it is straightforward to observe that  
\[
C_{g_*} = P(C_g^{-1})^{[m]} P, \quad \text{where} \quad P = \operatorname{circ}(0,0,\dots,1),
\]  
so we can derive the result by considering the reciprocal polynomial of \( g \langle X \rangle \).  

\begin{theorem}  
	Let \( g \langle X \rangle \in \mathbb{F}_{q^{2m}}[x,\theta] \) be a monic polynomial of degree \( m \) with \( g(0) \neq 0 \), and let \( \operatorname{ord}(g) = n \geq 2 \). Consider the map \( \theta: \mathbb{F}_{q^{2m}} \to \mathbb{F}_{q^{2m}} \) defined by \( \theta(a) = a^q \). If \( g \langle X \rangle \) generates recursively involutory MDS matrices, then its reciprocal polynomial \( g_* \langle X \rangle \) also generates recursively involutory MDS matrices.  
\end{theorem}  

\begin{proof}
	The proof is supported by the evidence presented in Theorem \ref{5thm4.10}. 
\end{proof}

To better understand the relationship between new MDS matrices and existing structures, we first introduce the Hadamard product. This operation will play a crucial role in our proof of results related to new MDS matrices derived from established ones.
\begin{definition}\cite[Definition 1.3]{Render1995Algebras}
	Let \( P  \langle X \rangle \) and \( Q  \langle X \rangle \) be two polynomials in $\mathbb{F}_q[X;~\theta]$ such that
	\[
	P  \langle X \rangle = \sum_{i=0}^{m} a_i X^i,
	\]
	\text{and}
	\[
	Q  \langle X \rangle = \sum_{i=0}^{n} b_i X^i.
	\]
	
	\noindent Then, the Hadamard product \( h\langle X \rangle \) of \( P  \langle X \rangle \) and \( Q  \langle X \rangle \) is given by
	\[
	h\langle X \rangle = P  \langle X \rangle \diamond Q  \langle X \rangle,
	\]
	where \( \diamond \) denotes the Hadamard product and the polynomial \( h\langle X \rangle \) is defined by its coefficients as follows:
	\[
	h\langle X \rangle = \sum_{i=0}^{\min(m,n)} c_i X^i,
	\]
	where
	\[
	c_i = a_i b_i \text{ for } i = 0, 1, \ldots, \min(m,n).
	\]

	
	\noindent The Hadamard \( s \)-power of \( P  \langle X \rangle \) is denoted by \( P  \langle X \rangle^{\diamond s} \), and  defined as
	\[
	P  \langle X \rangle^{\diamond s} = \sum_{i=0}^{m} (a_i)^s X^i.
	\]
	
\end{definition}
\begin{theorem}\label{5t11}
	Let $g  \langle X \rangle$ be a polynomial of degree $m$ over $\mathbb{F}_{q^{2m}}[X,\theta]$ and suppose $g  \langle X \rangle$ gives a quasi recursive MDS matrix. Then, $(g  \langle X \rangle)^{\diamond2^t}$ also gives a quasi recursive MDS matrix, where $2^t\leq q^{2m}-1.$
\end{theorem}
\begin{proof}
	Let $g  \langle X \rangle=g_0+g_1X+\cdots+g_{m-1}X^{m-1}+X^m$ $\in$ $\mathbb{F}_{q^{2m}}[X,\theta]$  gives quasi recursive MDS matrix. Then, this implies that 
	\begin{eqnarray}
		N_g&=&C_g^{[m-1]}\cdot C_g^{[m-2]}\cdots C_g^{[1]}\cdot C_g,
	\end{eqnarray}
	is an MDS matrix. Take a Hadamard product of $g  \langle X \rangle$ for $2^t\leq q^{2m}-1,$
	$$h\langle X \rangle=(g  \langle X \rangle)^{\diamond^{2^t}}=g^{2^t}_0+g^{2^t}_1X+g^{2^t}_2X^2+\cdots+g^{2^t}_{m-1}X^{m-1}+X^m.$$
	Suppose
	\begin{eqnarray*}
		N_g&=&C_g^{[m-1]}\cdot C_g^{[m-2]}\cdots C_g^{[1]}\cdot C_g\\
		&=&\begin{bmatrix}
			a_{11}&a_{12}&\cdots&a_{1m}\\
			a_{21}&a_{22}&\cdots&a_{2m}\\
			\vdots&\vdots&\ddots&\vdots\\
			a_{m1}&a_{m2}&\cdots&a_{mm}\\
		\end{bmatrix}_{m\times m}.
	\end{eqnarray*}
	Let quasi recursive matrix defined by the polynomial $h\langle X \rangle$ be
	\begin{eqnarray*}
		N_h&=&C_h^{[m-1]}\cdot C_h^{[m-2]}\cdots C_h^{[1]}\cdot C_h\\
		&=&\begin{bmatrix}
			0&1&\cdots&0\\
			0&0&\cdots&0\\
			\vdots&\vdots&\ddots&\vdots\\
			0&0&\cdots&1\\
			\theta^{m-1}(g^{2^t}_0)&\theta^{m-1}(g^{2^t}_1)&\cdots&\theta^{m-1}(g^{2^t}_{m-1})
		\end{bmatrix}_{m\times m}\cdot\begin{bmatrix}
			0&1&\cdots&0\\
			0&0&\cdots&0\\
			\vdots&\vdots&\ddots&\vdots\\
			0&0&\cdots&1\\
			\theta^{m-2}(g^{2^t}_0)&\theta^{m-2}(g^{2^t}_1)&\cdots&\theta^{m-2}(g^{2^t}_{m-1})
		\end{bmatrix}_{m\times m}\\
		&~&\cdots
		\begin{bmatrix}
			0&1&\cdots&0\\
			0&0&\cdots&0\\
			\vdots&\vdots&\cdots&\vdots\\
			0&0&\ddots&1\\
			\theta(g^{2^t}_0)&\theta(g^{2^t}_1)&\cdots&\theta(g^{2^t}_{m-1})
		\end{bmatrix}_{m\times m}\cdot \begin{bmatrix}
			0&1&\cdots&0\\
			0&0&\cdots&0\\
			\vdots&\vdots&\cdots&\vdots\\
			0&0&\ddots&1\\
			g^{2^t}_0&g^{2^t}_1&\cdots&g^{2^t}_{m-1}
		\end{bmatrix}_{m\times m}.
	\end{eqnarray*}
	Take any submatrix of $N_h$ of order $s$, say $N_{h_s}$
	$$N_{h_s}=\begin{bmatrix}
		a'_{i_1j_1}&a'_{i_1j_2}&\cdots&a'_{i_1j_s}\\
		a'_{i_2j_1}&a'_{i_2j_2}&\cdots&a'_{i_2j_s}\\
		\vdots&\vdots&\ddots&\vdots\\
		a'_{i_sj_1}&a'_{i_sj_2}&\cdots&a'_{i_sj_s}\\
	\end{bmatrix}_{m\times m}.$$
	The matrix $N_{h_s}$ corresponds to some submatrix of $N_g$ of order $s$, say $N_{g_s}$ such that
	$$\det(N_{h_s})=(\det(N_{g_s}))^{2^t}.$$ Since $N_{g}$ is an MDS matrix, the above relation implies that $\det(N_{h_s})\neq 0.$ Hence, $N_h$ is an MDS matrix.
\end{proof}
If we take  $\theta=I|_{\mathbb{F}_{2^{2m}}}$ in Theorem \ref{5t11}, then we obtain the following corollary:
\begin{corollary}\label{5cor5}
	Let $g  \langle X \rangle$ be a polynomial over $\mathbb{F}_{q^m}[X]$ and suppose $g  \langle X \rangle$ gives a recursive MDS matrix$,i.e.,$ $C^m_g$ is MDS. Then, $(g  \langle X \rangle)^{\diamond2^t}$ gives quasi recursive MDS matrix for  $2^t\leq q^m-1.$
\end{corollary}

The Hadamard product allows us to construct many different quasi recursive MDS matrices while preserving the MDS property. 
This provides greater flexibility in selecting diffusion matrices with desirable structure and efficient implementation over $\mathbb{F}_q$. To illustrate this idea, we present the following example, which demonstrates how new quasi recursive MDS matrices can be obtained via the Hadamard product.

\begin{example}
	Consider the finite field $\mathbb{F}_{2^8}=\frac{\mathbb{F}_2[X]}{( \mu  ( X ))}$, where $\mu  \langle X \rangle=X^8+X^4+X^3+X^2+1$ is a primitive polynomial over $\mathbb{F}_2$. Next, let $\alpha\in \mathbb{F}_{2^8}$ be a root of $\mu( X)$, i.e., a primitive element of $\mathbb{F}_{2^8}$. Suppose a polynomial $g  ( X)=X^4+\alpha^{76}X^3+\alpha^{251}X^2+\alpha^{81}X+\alpha^{10}.$ Then $C_g=Companion(\alpha^{10},\alpha^{81},\alpha^{251},\alpha^{76})$ and $C^4_{g}$ is an MDS matrix by \cite[Example 2]{gupta2017towards}. By using Corollary \ref{5cor5}, we found 7 polynomials with the help of Hadamard product of $g(X)$, which yields 7 new recursive MDS matrices. These seven polynomials are:
	\begin{eqnarray*}
		g_1  ( X )=g  ( X )^{\diamond2^1}&=& X^4+\alpha^{152}X^3+\alpha^{247}X^2+\alpha^{162}X+\alpha^{20},\\
		g_2  ( X )=g  ( X )^{\diamond2^2}&=& X^4+\alpha^{49}X^3+\alpha^{239}X^2+\alpha^{69}X+\alpha^{40},\\
		g_3 ( X )=g  ( X )^{\diamond2^3}&=& x^4+\alpha^{98}X^3+\alpha^{223}X^2+\alpha^{138}X+\alpha^{80},\\
		g_4 ( X )=g  ( X )^{\diamond2^4}&=& X^4+\alpha^{196}X^3+\alpha^{191}X^2+\alpha^{138}X+\alpha^{160},\\
		g_5  ( X )=g  ( X )^{\diamond2^5}&=& X^4+\alpha^{137}X^3+\alpha^{127}X^2+\alpha^{21}X+\alpha^{65},\\	
		g_6  ( X )=g  ( X )^{\diamond2^6}&=& X^4+\alpha^{19}X^3+\alpha^{254}X^2+\alpha^{42}X+\alpha^{130},\\
		g_7 ( X )=g  ( X )^{\diamond2^7}&=& X^4+\alpha^{38}X^3+\alpha^{253}X^2+\alpha^{84}X+\alpha^{5},
	\end{eqnarray*}
	and the recursive MDS matrices associated with above polynomials are as follows:
	\begin{eqnarray*}
		C_{g_1}=\begin{bmatrix}
			\alpha^{20}&\alpha^{162} &\alpha^{247} &\alpha^{152} \\ 
			\alpha^{172}&\alpha^{126} &\alpha^{129} &\alpha^{75} \\
			\alpha^{95}&\alpha^{79} &\alpha^{149} &\alpha^{13} \\ 
			\alpha^{33}&\alpha^{8} &\alpha^{108} &\alpha^{39} \\ 
		\end{bmatrix}_{4\times 4}&,&	C_{g_2}=\begin{bmatrix}
			\alpha^{40}&\alpha^{69} &\alpha^{239} &\alpha^{49} \\ 
			\alpha^{89}&\alpha^{252} &\alpha^{3} &\alpha^{150} \\
			\alpha^{190}&\alpha^{158} &\alpha^{43} &\alpha^{26} \\ 
			\alpha^{66}&\alpha^{16} &\alpha^{216} &\alpha^{78} \\ 
		\end{bmatrix}_{4\times 4},\\
		C_{g_3}=\begin{bmatrix}
			\alpha^{80}&\alpha^{138} &\alpha^{223} &\alpha^{98} \\ 
			\alpha^{178}&\alpha^{249} &\alpha^{6} &\alpha^{45} \\
			\alpha^{125}&\alpha^{61} &\alpha^{86} &\alpha^{52} \\ 
			\alpha^{32}&\alpha^{32} &\alpha^{177} &\alpha^{156} \\ 
		\end{bmatrix}_{4\times 4}&,&	C_{g_4}=\begin{bmatrix}
			\alpha^{160}&\alpha^{21} &\alpha^{191} &\alpha^{196} \\ 
			\alpha^{101}&\alpha^{243} &\alpha^{12} &\alpha^{90} \\
			\alpha^{250}&\alpha^{79} &\alpha^{149} &\alpha^{13} \\ 
			\alpha^{33}&\alpha^{8} &\alpha^{108} &\alpha^{39} \\ 
		\end{bmatrix}_{4\times 4},\\
		C_{g_5}=\begin{bmatrix}
			\alpha^{65}&\alpha^{42} &\alpha^{127} &\alpha^{137} \\ 
			\alpha^{202}&\alpha^{231} &\alpha^{24} &\alpha^{180} \\
			\alpha^{245}&\alpha^{244} &\alpha^{89} &\alpha^{208} \\ 
			\alpha^{18}&\alpha^{128} &\alpha^{198} &\alpha^{114} \\ 
		\end{bmatrix}_{4\times 4}&,&	C_{g_6}=\begin{bmatrix}
			\alpha^{130}&\alpha^{84} &\alpha^{254} &\alpha^{19} \\ 
			\alpha^{149}&\alpha^{207} &\alpha^{48} &\alpha^{105} \\
			\alpha^{235}&\alpha^{233} &\alpha^{178} &\alpha^{161} \\ 
			\alpha^{36}&\alpha^{1} &\alpha^{141} &\alpha^{228} \\ 
		\end{bmatrix}_{4\times 4}\\
		C_{g_7}&=&\begin{bmatrix}
			\alpha^{5}&\alpha^{168} &\alpha^{253} &\alpha^{38} \\ 
			\alpha^{43}&\alpha^{159} &\alpha^{96} &\alpha^{210} \\
			\alpha^{215}&\alpha^{211} &\alpha^{101} &\alpha^{67} \\ 
			\alpha^{72}&\alpha^{2} &\alpha^{27} &\alpha^{201} \\ 
		\end{bmatrix}_{4\times 4}.\\
	\end{eqnarray*}
\end{example}

\section{\textbf{Conclusion and future work}}

The construction of $\delta_{\theta}$-circulant MDS matrices and quasi recursive MDS matrices studied with the help of skew polynomial ring. Some examples of such constructions are presented. Note that if we take $\delta=0$ and $\theta=I_{\mathbb{F}_q},$ then the matrix we have obtained cover those obtained in \cite{cauchois2019circulant} and \cite{cauchois2016direct}, respectively. An important requirement relevant to the involutory MDS matrices, we obtain quasi recursive MDS matrices which is involutory in Theorem \ref{5t1}.

Although quasi recursive MDS matrices are well known, our choice of the automorphism parameter yields an involutory MDS matrix, leading to a computationally efficient inverse. This motivates further study of structured MDS matrices; in particular, the construction of quasi recursive MDS matrices via $\theta$-derivations remains open. We also leave a complete characterization of MDS $\delta_{\theta}$-circulant matrices as an open problem (see Open Problem~\ref{op:dtheta}). Moreover, it would be worthwhile to investigate whether $\delta_{\theta}$-circulant matrices possess orthogonality properties. Additionally, future work may explore whether these constructions can be extended to produce semi-involutory or semi-orthogonal MDS matrices.


\section{\bf Declarations}

\noindent \textbf{Funding}\newline
This research did not receive any specific grant from funding agencies in the public, commercial, or not-for-profit sectors.\newline

\noindent \textbf{Data Availability Statement}\newline
Data sharing is not applicable to this article as no data sets were
generated or analyzed during the current study.\newline

\noindent \textbf{Conflicts of Interest}\newline
The authors declare that they have no conflicts of interest. \newline



\end{document}